\documentclass[11pt]{article}
\usepackage[utf8]{inputenc}
\usepackage{fullpage}
\usepackage{amsmath, mathtools, amssymb, amsfonts, amsthm}
\usepackage[english]{babel}
\usepackage{color}
\usepackage[margin=1in]{geometry}
\usepackage{datetime}
\usepackage{enumitem}
\usepackage[compact]{titlesec}
\usepackage{nicefrac}
\usepackage{mathdots,url}
\usepackage{graphicx}
\usepackage{subfigure}
\usepackage{dsfont,bm}
\usepackage{thmtools} 
\usepackage{thm-restate}
\usepackage{calc}
\usepackage{accents}
\usepackage{tikz}
\usepackage{caption}

\usepackage[numbers]{natbib}

\newtheorem{theorem}{Theorem}[section] 
\newtheorem{claim}[theorem]{Claim}
\newtheorem{proposition}[theorem]{Proposition} 
\newtheorem{lemma}[theorem]{Lemma} 
\newtheorem{corollary}[theorem]{Corollary}
\newtheorem{conjecture}[theorem]{Conjecture}

\theoremstyle{definition}

\newtheorem{definition}[theorem]{Definition}
\newtheorem{remark}[theorem]{Remark}

\newtheorem{assumption}{Assumption}

\DeclareMathOperator*{\argmax}{arg\,max}
\DeclareMathOperator*{\nsw}{NSW} 
\DeclareMathOperator{\opt}{OPT} 

\newcommand{\R}{{{\mathbb{R}}}}

\newcommand{\N}{{{\mathbb{N}}}}
\newcommand{\Z}{{{\mathbb{Z}}}}
\newcommand{\1}{{{\mathds{1}}}}

\newcommand{\supp}{{{\mathrm{supp}}}}
\newcommand{\nat}{{{$\sp{\natural}$}}}

\newcommand{\val}{v} 
\newcommand{\f}{r} 
\newcommand{\G}{{\cal G}} 
\newcommand{\T}{\ensuremath{V}} 
\newcommand{\E}{E} 
\newcommand{\Se}{S} 
\newcommand{\M}{{\cal M}} 
\newcommand{\I}{{\cal I}} 
\newcommand{\co}{c} 
\newcommand{\Hs}{{\cal H}} 
\newcommand{\F}{{\cal L}} 
\newcommand{\A}{{\cal A}} 
\newcommand{\B}{B}

\newcommand{\ls}{Y}

\newcommand{\s}{k}
\newcommand{\g}{\gamma}
\newcommand{\kc}{{c}} 
\newcommand{\pr}{{q}}
\newcommand{\y}{y} 
\newcommand{\tcnt}{t} 

\newcommand{\ry}{\y^r}
\newcommand{\de}{d} 
\newcommand{\De}{D} 

\usepackage[colorlinks=true, allcolors=blue]{hyperref}
\newenvironment{claimproof}[1][\proofname]
	{  
        \proof[Proof of Claim]
        
	}
	{      
		\endproof
	}
\makeatletter
\newenvironment{tagequation}[1][Missing Tag]{
    \begin{equation}
    \tag{#1}
    \begin{aligned}
    \hphantom{#1}
    \end{aligned}
    \begin{aligned}
}{
    \end{aligned}
    \end{equation}
}
\makeatother
\title{Approximating Nash Social Welfare under Rado Valuations}

\author{Jugal Garg\thanks{
University of Illinois at Urbana-Champaign. Supported by NSF Grant CCF-1942321 (CAREER)}\\ \texttt{jugal@illinois.edu}
\and
Edin Husi\' c\thanks{Department of Mathematics, 
	London School of Economics and Political Science. This project has received funding from the European Research Council (ERC) under the European Union's Horizon 2020 research and innovation programme (grant agreement no. ScaleOpt--757481).}\\ \texttt{e.husic@lse.ac.uk}
\and
L\' aszl\' o A. V\' egh\footnotemark[2]\\ \texttt{l.vegh@lse.ac.uk}
}
\date{}

\begin{document}

\maketitle
\thispagestyle{empty}

\begin{abstract}
We consider the problem of approximating maximum Nash social welfare (NSW) while allocating a set of indivisible items to $n$ agents. The NSW is a popular objective that provides a balanced tradeoff between the often conflicting requirements of fairness and efficiency, defined as the weighted geometric mean of agents' valuations. For the symmetric additive case of the problem, where agents have the same weight with additive valuations, the first constant-factor approximation algorithm was obtained in 2015. This led to a flurry of work obtaining constant-factor approximation algorithms for the symmetric case under mild generalizations of additive, and $O(n)$-approximation algorithms for more general valuations and for the asymmetric case. 

In this paper, we make significant progress towards both symmetric and asymmetric NSW problems. We present the first constant-factor approximation algorithm for the symmetric case under \emph{Rado} valuations. Rado valuations form a general class of valuation functions that arise from maximum cost independent matching problems, including as special cases assignment (OXS) valuations and weighted matroid rank functions. Furthermore, our approach also gives the first constant-factor approximation algorithm for the asymmetric case under Rado valuations, provided that the maximum ratio between the weights is bounded by a constant. 
\end{abstract}
\newpage
\tableofcontents
\thispagestyle{empty}

\newpage
\section{Introduction}
Fair and efficient allocation of resources is a fundamental problem in many disciplines, including computer science, economics, and social choice theory; see, e.g., several excellent books written specifically on this problem~\cite{Barbanel,Brams1996,BrandtCELP16,Moulin2004,Robertson1998,R16,Young1995}. The Nash social welfare (NSW) is a popular objective that provides a balanced tradeoff between the often conflicting requirements of fairness and efficiency in contrast to the other popular social welfare concepts, including the \emph{utilitarian social welfare} and the \emph{max-min fairness}, also known as the \emph{Santa Claus} problem. It is no wonder that it was discovered independently in several different contexts: First, as the unique solution to a bargaining game by Nash in 1950~\cite{Kaneko1979,nash1950bargaining}. It also coincides with the popular notion of competitive equilibrium with equal incomes (CEEI) in economics~\cite{varian1973equity}, and as a notion of proportional fairness in networking~\cite{kelly1997charging}. 

In the discrete NSW problem, one needs to allocate a set $\G$ of $m$ indivisible items to a set $\A$ of $n$ agents where each agent $i$ has a valuation function $v_i: 2^{\G}\to \R_+$ and weight (entitlement) $w_i>0$. The goal is to find an allocation maximizing the NSW, defined as the weighted geometric mean of the valuations:

\begin{equation}\label{eq:nash-welfare}
\max\left\{\left(\prod_{i\in \A} \val_i(S_i)^{w_i} \right)^{1/\sum_{i\in \A} w_i}:\, \{S_i: i\in \A\} \textrm{ forms a partition of }\G \right\}\, .
\end{equation}
We refer to the special case when all agents have equal weight (i.e., $w_i=1$) as the \emph{symmetric} NSW problem, and call the general case the \emph{asymmetric} NSW problem. While the early works only considered the symmetric NSW, the asymmetric case has also been well-studied since the seventies \cite{harsanyi1972generalized,kalai1977nonsymmetric}, and has 
found many applications in different areas, such as bargaining theory~\cite{chae2010bargaining,Laruelle2007}, water resource allocation~\cite{degefu2016water,houba2013asymmetric}, and climate agreements~\cite{yu2017nash}. 
Another distinctive feature of NSW is invariance under scaling of the valuation functions, i.e., one obtains the same optimal partition even if some agents over- or under-report their valuations by a constant factor.

\paragraph{Computational complexity}
Finding an exact solution to the NSW problem is NP-hard even for two identical agents with additive valuations: observe that the partition problem reduces to the NSW problem. Moreover, the problem is NP-hard to approximate within a factor better than $1.069$~\cite{garg2017satiation} for additive valuations, and better than $1.5819$~\cite{GargKK20} for submodular valuations. These results hold already in the symmetric case. 

For the symmetric NSW problem with additive valuations, Cole and Gkatzelis~\cite{cole2015approximating}, in a breakthrough result, provided the first constant-factor approximation algorithm using an approach based on \emph{spending-restricted} market equilibrium, whose analysis was later improved in~\cite{cole2017convex}. Anari et al.~\cite{anari2017nash} provided another approach using the theory of real stable polynomials. Barman et al.~\cite{barman2018finding} developed yet another approach based on local search that provides the state-of-the-art approximation factor of $1.45$.

These approaches have also been extended to obtain constant-factor approximation algorithms for mild generalizations of additive, namely, budget-additive~\cite{garg2018approximating}, separable piecewise linear concave (SPLC)~\cite{anari2018nash}, and their combination budget-SPLC~\cite{ChaudhuryCGGHM18} valuations. All these approaches heavily exploit the symmetry of agents and the characteristics of `additive-like' valuations, such as the notion of a maximum bang-per-buck (MBB) items, which make them hard to extend to significantly more general settings.

For more general valuations or the asymmetric NSW problem, new approaches~\cite{barman2020tight,chaudhury2020fair,GargKK20} have been recently developed, resulting in the state-of-the-art factor of $O(n)$ for the asymmetric Nash problem under subadditive valuations. However, their analysis is based on averaging arguments, making them hard to yield a factor better than $O(n)$ even for the special cases, e.g., OXS valuations, or only two types of agents with weights 1 or 2 under additive valuations. Therefore, $O(n)$ remained the best approximation factor for the symmetric NSW problem beyond `additive-like' valuations or for the asymmetric NSW problem. 

\paragraph{Our contributions}
We make significant progress towards both symmetric and asymmetric NSW problems. Firstly, we obtain a constant-factor approximation for a broad class of submodular valuations we call \emph{Rado} valuations. This is a common generalization of OXS valuations and weighted matroid functions. A Rado valuation of an agent $i\in \A$ is specified by a bipartite graph $(\G, \T_i; \E_i)$, edge costs $\co_i : \E_i \to \R_+$ and a matroid $\M_i = (\T_i, \I_i)$. The value $\val_i(S)$ of a subset of items $S\subseteq \G$ is given as the maximum cost of a matching  between nodes in $S$ and nodes in $\T_i$ such that the endpoints in  $\T_i$ form an independent set in the matroid $\M_i$. Relation between popular classes of valuations functions follows~\cite{DBLP:journals/geb/LehmannLN06,nisan2007algorithmic}:
\[
\text{Additive} \subsetneq \text{SPLC} \subsetneq 
\begin{array}{c}
\text{OXS} \\
\text{Weighted Matroid Rank}
\end{array}
\subsetneq \text{Rado} \subsetneq \text{GS} \subsetneq \text{Submodular} \subsetneq \text{XOS} \subsetneq \text{Subadditive} \, . 
\]

\begin{restatable}{theorem}{mainSym}\label{thm:main-symmetric}
There exists a polynomial-time $256 e^{3/e}{\approx }772$-approximation algorithm for the symmetric Nash social welfare problem under Rado valuations. 
\end{restatable}

Rado valuations form a subclass of gross substitutes (GS) valuations. In fact, it was conjectured by Frank in 2003 that every GS valuation arises as a Rado valuation, see Section~\ref{section:utility}. In Section~\ref{sec:conjectures} we give a counterexample and formulate a slight strengthening of this conjecture.
\medskip

Secondly, we obtain a constant-factor approximation for the asymmetric NSW problem under Rado valuations, provided that the maximum ratio between the weights is bounded by a constant.
Assume the weights $w_i$ of the agents fall in the interval $[1, \g-1]$ for some $\g\ge2$. 

\begin{restatable}[Main]{theorem}{main}\label{thm:main}
There exists a polynomial-time $256\g^3$-approximation algorithm for the Nash social welfare problem with Rado valuation functions. For additive valuation functions, there exists a polynomial-time $16\g$-approximation algorithm.\footnote{We note that $\g$ in the theorem can be replaced by $\min\left\{O\left(\frac{\g}{\log \g} \right), n \right\}$ as we show in Section~\ref{subsection:upperBoundsPreliminaries} and Section~\ref{section:productBounds}.}
\end{restatable}

We note that even if the weights of the agents are bounded, an $O(1)$-approximation for the symmetric case does not yield an $O(1)$-approximation to the asymmetric case.\footnote{To illustrate this point, consider two items and two agents with weights $w_1=2$, $w_2=1$ and additive valuations $\val_{1}(\{1\})=M$, $\val_{1}(\{2\})=1$, $\val_{2}(\{1\})=M+1$, $\val_{2}(\{2\})=1$, and so on, where $M$ is an arbitrarily large number. The unique optimal solution to the symmetric case (by setting $w'_1=w'_2=1$) is allocating good 2 to agent 1 and good 1 to agent 2. However, this returns an NSW value $(M+1)^{1/3}$ for the original weights. This can be worse by an arbitrary factor than the value $M^{2/3}$ obtainable by assigning good 1 to agent 1 and good 2 to agent 2.}
 Table~\ref{table:results} summarizes the updated best approximation guarantees for the problem under various valuation functions. 
\begin{table}[h]
\centering
\begin{tabular}{|c|c|c|}
\hline
\bf{Valuations}         & \bf{Symmetric}                                    & \bf{Asymmetric} \\ \hline \hline
Additive                & 1.45 \cite{barman2018finding} & $O(\g)$ [Theorem~\ref{thm:main}] \\ \hline
SPLC                    & 1.45 \cite{ChaudhuryCGGHM18}                      &  $O(\g^3)$ [Theorem~\ref{thm:main}]\\ \hline
Rado                    & $O(1)$ [Theorem~\ref{thm:main-symmetric}] & $O(\g^3)$ [Theorem~\ref{thm:main}]  \\ \hline
Subadditive             & $O(n)$~\cite{barman2020tight,chaudhury2020fair}                   & $O(n)$~\cite{barman2020tight,chaudhury2020fair}  \\ \hline
\end{tabular}
\captionsetup{singlelinecheck=off}
\caption[.]{Summary of the best approximation algorithms for the NSW problem.
Definitions of valuations functions are deferred to Section~\ref{section:valuationFunctions}.}
\label{table:results}
\end{table}

\subsection{Main ideas}
Our approach is based on a mixed-integer programming relaxation, using a careful combination of convex programming relaxations and combinatorial arguments.

The NSW problem is given with discrete valuation functions $\val_i:2^\G\to \R_+$. In order to apply convex programming techniques, we first need to obtain a convex programming relaxation; already this turns out to be a nontrivial task. As explained in Section~\ref{section:utility}, gross substitute valuations are the subclass of discrete valuations where a concave extension can be naturally defined.

Already for  additive valuations,  the natural relaxation of the NSW problem has unbounded integrality gap~\cite{cole2015approximating}. In order to formulate a mixed integer program, we identify a set $\Hs$ of $n$ items, and require that all these items must be integrally allocated. We do not know if this relaxation can be solved in polynomial time: we only provide an approximate solution to a further relaxation.

For the set $\Hs$, we aim to identify the set of the `most important' items. We find the allocation maximizing the NSW value assuming each agent can obtain just a single item, and select $\Hs$ as the set of the items selected in this allocation. This can be efficiently solved as a maximum weight matching problem. The algorithm in \cite{GargKK20} also starts with such a matching. One cannot commit to assigning these items to the agents, as it may result in an arbitrary bad outcome; the approach in  \cite{GargKK20} is an intricate combinatorial scheme with iterated matchings and reallocations to obtain an $O(n\log n)$ approximation for submodular valuations. 
Our result implies that the mixed integer relaxation that requires $\Hs$ to be integrally allocated has a constant integrality gap, in contrast to the standard continuous relaxation. As a possible explanation why this happens, we make a connection to the approach of \citet{cole2015approximating} in Section~\ref{section:connectionToSR}, showing that all `expensive' items in the spending restricted equilibrium will be included in $\Hs$.
\medskip

We give a detailed exposition of the overall approach  and formulate the main lemmas in Section~\ref{sec:appraoch}, split into five phases. Here, we only give a high-level overview. \ref{phase1} selects $\Hs$ as above. \ref{phase2} approximates the mixed relaxation by another mixed integer program~\eqref{prog:decomposedUtility} that assigns items $\G\setminus \Hs$ fractionally to the agents, 
and at most one item from $\Hs$ to each agent. This is not a relaxation of the original problem anymore, as an optimal solution may allocate multiple items from $\Hs$ to the same agent. However, \eqref{prog:decomposedUtility} approximates the original mixed 
within a factor  $\g$. 
We note that this is the only part of our reductions that depends on the bound $\gamma$.

Solving \eqref{prog:decomposedUtility} still does not turn out to be easy. In  \ref{phase3}, we find a 2-approximate solution by first solving the  restriction to $\G\setminus \Hs$---a convex program---then optimally assigning the items in $\Hs$.

All reductions thus far work for general subadditive valuations. In \ref{phase4} we exploit combinatorial  properties of the  concave extension of Rado valuations to obtain a sparse solution. We first show that the restriction of~\eqref{prog:decomposedUtility} to $\G\setminus \Hs$ has a basic optimal solution 
with at most $|\A|+2|\G \setminus \Hs|$ non-zero variables. 
We note that this is on its own an interesting new  \emph{rational convex program} \cite{Vazirani2012}, the first example we are aware of with an exponential number of constraints, given by a separation oracle.
We then further sparsify the solution to 
at most $2|\A|+|\G \setminus \Hs|$ non-zero variables, at the expense of losing at most half of the objective value. 

At this point, we have a mixed integer solution that is not too far from an integral one. 
Namely, $\Hs$ is already allocated integrally and 
 $\G\setminus \Hs$ is allocated to agents fractionally but with at most $2|\A|+|\G \setminus \Hs|$ non-zero variables. 
Thus, it suffices to fix a suitable subset of $2|\A|$ fractional variables to zero of the non-zero to obtain a feasible solution, and round the rest of the variables to 1.
However, this may not be viable for any subset.

In the final~\ref{phase5}, we make use of the choice of  $\Hs$ as the set of items allocated in the best allocation with one item per agent. Using this property, we carefully recombine the matching in the mixed assignment and the initial allocation of the items in $\Hs$ by swapping around alternating cycles. This enables the final rounding step to obtain an integer allocation.

\subsection{Further related work}
We briefly mention further results on Nash social welfare, utilitarian social welfare and max-min welfare.

\paragraph{Nash social welfare} 
NSW has turned out to be the focal point in fair division. Amenable fairness properties of NSW are underlined by~\citet{caragiannis2019unreasonable}, who call the solution  ‘unreasonably’ fair and efficient. The same paper introduces an algorithm for finding optimum NSW allocation, which is deployed on the website \href{http://www.spliddit.org}{spliddit.org} and used for fair allocation of indivisible goods~\cite{goldman2015spliddit}. Approximation algorithms for the NSW also preserve many nice fairness properties, as shown in~\cite{CaragiannisGH19,chaudhury2020fair,GargM20}. 

\paragraph{Utilitarian social welfare} In this setting, the goal is to find a partition of the items that maximizes the sum of agents' valuations. Note that this problem is straightforward for additive valuations. For gross substitutes valuations (see Definition~\ref{def:gs}), the optimal partition corresponds to a \emph{Walrasian equilibrium}: there exists a price vector such that each agent receives an optimal bundle at such prices. Such an allocation can be efficiently computed \cite{Gul1999,Kelso1982}.  G\"ul and Stachetti \cite{Gul1999} also showed that the converse is essentially true: if a class $\cal C$ of valuation functions contains all unit demand valuations, and there exists a Walrasian equilibrium for an arbitrary choice of valuation functions from $\cal C$, then $\cal C$ must be a subset of gross substitutes valuations.

For submodular valuations there is an $\frac{e}{e-1}\approx 1.5819$-approximation algorithm by~\citet{vondrak2008optimal} and this is the best possible~\cite{khot2008inapproximability}. \citet{feige2009maximizing} gave a 2-approximation algorithm for the social welfare problem under subadditive valuations assuming access to particular demand queries.

\paragraph{Max-min welfare} In this problem the objective is to maximize the minimum valuation of any agent.
This NP-hard problem can be seen as an absolute fairness problem and it has been appropriately named the Santa Claus problem~\cite{bansal2006santa}. It is a significant open problem to obtain a 
constant-factor approximation for additive valuations: such algorithms are known only for a restricted subclasses, see \citet{annamalai2017combinatorial,davies2020tale}. 
For additive (resp. submodular) valuations the best approximation factor is $O(\sqrt[3]{n} \log^3 n)$ by~\citet{asadpour2010approximation} (resp. $O(n)$ by~\citet{khot2007approximation}).

\paragraph{Organization of the paper} In Section~\ref{section:preliminaries} we formally define all the notation and concepts. 
Here, we also explain the significance of the gross substitutes and Rado valuations for the problem and our approach. 
Missing proofs from Section~\ref{section:preliminaries}, as well as a detailed discussion of Rado valuations 
are presented in Section~\ref{section:utilityDetails}.
In Section~\ref{sec:appraoch} we give a rigorous overview of the algorithm together with main lemmas proof ideas. 
Sections~\ref{sec:approximatingMatchingRelaxation}-\ref{section:rounding} contain more detailed arguments for the various phases. Section~\ref{section:connectionToSR} compares our approach with the spending restricted equilibria in \cite{cole2015approximating}. Concluding remarks are given in Section~\ref{sec:conclusions}.

\section{Preliminaries}
\label{section:preliminaries}
Throughout, we let $\G$ denote a finite set of $m$ indivisible items (goods), and $\A$ a set of $n$ agents. Each of the agents $i\in \A$ are equipped with a valuation function $\val_i:2^{\G}\to \R$. 
Throughout, we use the shorthand notation $\val_{ij}=\val_i(\{j\})$ to denote the valuation of agent $i$ for a single unit of item $j$. 

Given a subset $S \subseteq \G$ we will denote with $\chi_S$ the characteristic vector of $S$.
For $k\in\Z$, we let $[k]=\{1,2,\ldots,k\}$. 
A {\em bipartite graph} $(U,V;E)$ has node set $U\cup V$ and an undirected edge set $E\subseteq U\times V$. For an edge subset $F\subseteq E$, we let $\delta_U(E)$ and $\delta_V(E)$ denote the set of endpoints of $E$ in $U$ and in $V$, respectively.

A {\em matroid} on a finite ground set $\T$ is given as $\M=(\T,\I)$, where $\I\subseteq 2^\T$ is a nonempty collection of {\em independent sets}. This collection is required to satisfy the {\em independence axioms}:
\begin{enumerate}[label=(I\arabic*)]
\item\label{ax:mon} {\em Monotonicity:} if $X\in \I$ then $Y\in \I$ for all $Y\subseteq X$, and
\item\label{ax:ex} {\em Exchange property:} if $X,Y\in \I$, $|X|<|Y|$, then there exists an $y\in Y\setminus X$ such that $X\cup\{y\}\in \I$.
\end{enumerate}
The \emph{rank function} $\f_{\M}:2^\T\to \Z_+$ associated with the matroid $\M$ is defined with $\f_{\M}(X)$ denoting the size of the largest independent subset of $X\subseteq \T$. A fundamental property implied by \ref{ax:ex} is that every maximal independent set in $X$ has size $\f_{\M}(X)$. The value $\f_{\M}(\T)$ is called the rank of the matroid, and the maximal independent sets are called {\em bases}. 
A set $X\subseteq \T$ is in $\I$ if and only if $r(X)=|X|$.
We refer the reader to \cite[Part IV]{schrijver2003combinatorial} for matroids and their role in optimization.

\subsection{Valuation functions}
\label{section:valuationFunctions}
By a \emph{valuation function}, we mean a  function $\val: 2^{\G}\to \R$  with $\val(\emptyset)=0$. 
Let us start with two simple examples of valuations. 
The function $\val$ is an \emph{additive valuation} if $\val(S)=\sum_{j\in S} \val_{j}$,
 and a \emph{unit demand  valuation} if $\val(S)=\max_{j\in S} \val_j$ where $\val_j \in \R_+$ represents the value of item $j\in \G$.

We now define some basic properties.
A function $\val: 2^{\G}\to \R_+$  is \emph{monotone} if $\val(X)\le \val(Y)$ for any $X\subseteq Y\subseteq \G$, \emph{subadditive} if 
\begin{equation}\label{eq:subadd}
\val(X)+\val(Y)\ge \val(X\cup Y)\quad \forall X,Y\subseteq \G\, ,
\end{equation}
and \emph{submodular} if
\begin{equation}\label{eq:submod}
\val(X)+\val(Y)\ge\val(X\cap Y)+\val(X\cup Y)\quad \forall X,Y\subseteq \G\, .
\end{equation}
Additive valuations and unit demand valuations satisfy all the above properties. 
Another basic example of submodular functions is the rank function $\f_\M$ of a matroid $\M=(V,\I)$. In fact, every integer valued monotone submodular set function on $V$ with $\val(X)\le |X|$ arises as the rank function of a matroid. Given a weighting $g\in\R^V$, the \emph{weighted rank function} $\f_g(X)$ is the maximum $g$-weight of a maximal independent set in $X$; this function is also submodular.

\paragraph{Gross substitute valuations} For a price vector $p\in \R^\G$ and a subset $S\subseteq \G$, we let $p(S)=\sum_{j\in S} p_j$. For a valuation function $\val: 2^{\G}\to \R_+$, the utility obtainable at prices $p$ from a set $S\subseteq\G$ is $\val(S)-p(S)$.
The  \emph{demand correspondence} is defined as 
\[
D(\val,p):=\argmax_{S\subseteq \G} \val(S)-p(S)\, .
\]
An important class of valuation functions is \emph{gross substitutes valuations}, defined by Kelso and Crawford in 1982 \cite{Kelso1982}:
\begin{definition}\label{def:gs}
The  valuation function  $\val: 2^{\G}\to \R_+$ is a  \emph{gross substitutes (GS) valuation} if for any $p,p'\in \R^{\G}$ such that $p'\ge p$ and any $S\in D(\val,p)$, there exists an $S'\in D(\val,p')$ such that $S\cap\{j: p_j=p'_j\}\subseteq S'$.
\end{definition}
That is, if we have an optimal bundle at prices $p$ and increase some of the prices, then there will be an optimal bundle that contains all items whose price remained unchanged. For a comprehensive survey on GS valuations, we refer the reader to the survey by Paes Leme
\cite{Leme2017}.

G\"ul and Stachetti~\cite{Gul1999} showed that every gross substitutes valuation is submodular. 
It turns out that gross substitute functions are intimately connected to  \emph{discrete convex analysis}, a general theory arising at the intersection of convex analysis and submodularity. 

Murota's book \cite{Murotabook} gives a comprehensive treatment of this field.
A central concavity concept on the integer lattice is that of \emph{M\nat-concave functions}. The definition specialized for valuation functions (corresponding to the sublattice $\{-\infty,0\}^\G$) is as follows.
\begin{definition}\label{def:M-conc}
The function  $\val: 2^{\G}\to \R_+$ is an \emph{M\nat-concave} if 
for any $X,Y\subseteq \G$ and $x\in X\setminus Y$,
\[
\val(X)+\val(Y)\le\max_{Z\subseteq Y\setminus X, |Z|\le 1} v((X\setminus\{x\})\cup Z)+v((Y\setminus Z)\cup\{x\})
\]
\end{definition}
That is, for any $x\in X\setminus Y$, the sum $\val(X)+\val(Y)$ is either non-decreasing if we move $x$ from $X$ to $Y$, or the sum is non-decreasing  by swapping $x$ for some $y\in Y\setminus X$. As established by Fujishige and Yang \cite{Fujishige2003}, these two concepts are equivalent:
\begin{theorem}[{\cite{Fujishige2003}}]\label{thm:gs-mnat}
The valuation function $\val: 2^{\G}\to \R_+$ is a gross substitutes valuation if and only if it is M\nat-concave.
\end{theorem}
This connection has enabled a fruitful interaction between the areas of mechanism design and discrete convexity, see e.g. \cite{murota2016discrete,Leme2017}.

\paragraph{Rado valuations} 
The key class of valuation functions for this paper will be 
 \emph{Rado valuation functions}, or \emph{Rado valuations}, defined next.
\begin{definition}\label{def:Rado}
Assume we are given a bipartite graph $(\G, \T; \E)$ with a cost function $\co : \E \to \R$ on the edges, and a matroid $ \M = (\T, \I)$.
For a subset of items $S \subseteq \G$, the \emph{Rado valuation function} $\val(S)$ is defined as the maximum cost of a matching $M$ in $(\G, \T; \E)$
such that $\delta_{\G}(M) \subseteq S$ and $\delta_{\T}(M) \in \I$, i.e.,
\begin{equation}\label{eq:Rado-def}
  \val(S) := \max\left\{\sum_{e\in M} \co(e) : M \text{ is a matching, } \delta_{\G}(M) \subseteq S, \delta_{\T}(M) \in \I\right\}\,.
\end{equation}
\end{definition}
We propose to name this class in honor of Richard Rado, who first studied the independent matching problem \cite{Rado1942}.

Let us consider the special case where the matroid $\M$ is the free matroid on $\T$, i.e., $\I = 2^{\T}$. 
In this case, the matroid constraints $\delta_{\T}(M) \in \I$ are void. 
The value of a set $S$ it then 
the maximum cost matching in the bipartite subgraph induced by $S \cup \T$. 
Such valuations are called \emph{assignment valuations} by~\citet{shapley1962optimal}, 
and \emph{OXS valuations} by~\citet{DBLP:journals/geb/LehmannLN06}. 

As another example of Rado valuations,
consider the case where $\T$ is a copy of the set of items $\G$, 
with each $j\in \G$ having a corresponding $j'\in \T$, 
and let $\E=\{(j,j'):\, j\in \G\}$. 
Let $g:\G\to \R$, and $\co_{jj'}=g_j$ for all $j\in \G$, and let $\f$ be rank function of $\M$. 
In this case the $\val(S)$ equals the weighted matroid rank function $\f_g(S)$. 

Assignment valuations and weighted matroid rank functions are well-known examples of M\nat-concave (and, according to Theorem~\ref{thm:gs-mnat}, gross substitutes) functions. 
We show that this is true in general for Rado valuations.

\begin{restatable}{lemma}{RadoIsMconcave}
\label{lemma:RadoIsMconcave}
Every Rado valuation $v:2^\G \to \R$ is an M\nat-concave function.
\end{restatable}

The proof is given in Section~\ref{section:utilityDetails}, using a more general construction by Murota \cite{Murotabook}. 
It is worth noting that in 2003, Frank posed the question on whether the converse is also true: is the class of M\nat-concave functions the same as those of Rado valuations?\footnote{Personal communication by Andr\'as Frank. See also Kazuo Murota's lecture~\cite{lecture}, the problem sheet~\cite{problemSheet}, and Renato Paes Leme's lecture~\cite{lectureRenato}.}
In Section~\ref{sec:conjectures} we use an example from \cite{DBLP:journals/geb/LehmannLN06}  showing that this is not the case. The main underlying reason is that this class is not minor closed. We then formulate a stronger conjecture, and mention an earlier conjecture by Ostrovsky and Paes Leme
\cite{ostrovsky2015gross}, partially refuted by Tran \cite{Tran2019}.

\subsection{Continuous valuation functions}
\label{section:utility}
The valuation functions $\val$ in the Nash social welfare problem are defined on subsets of $\G$. Our arguments are based on convex relaxations, which requires a continuous extension of the valuation functions to $\R_+^\G$. We provide such an extension for Rado valuations; however, we note that a suitable extension does not even exist for general submodular valuations.

By a \emph{continuous valuation function} we mean a continuous function $\val:[0,1]^\G\to \R$ with $\val(\mathbf{0})=0$. We slightly abuse the notation by using $\val$ to denote both  discrete and continuous valuations;  
the value of a subset $S\subseteq \G$ of items will be $\val(\chi_S)=\val(S)$.
Extending notions from discrete valuations, a function $f:\R_+^\G\to\R_+$ is \emph{monotone} if $f(x)\le f(y)$ for $x\le y$,  $x,y\in \R_+^{\G}$, and \emph{subadditive} if $f(x+y)\le f(x)+f(y)$ for any $x,y\in [0,1]^{\G}$ such that $x+y\in [0,1]^\G$.

Whereas our overall result requires the continuous extension of \emph{Rado valuations},  much weaker assumptions suffice for most parts of the argument, as formulated next. 
\begin{assumption}
\label{monotonicity}
  For every agent $i \in \A$ the continuous valuation function $\val_i:[0,1]^\G\to \R_+$ is  monotone, concave, and subadditive. 
\end{assumption}

\paragraph{Concave extensions of discrete valuations} 
For any discrete valuation function $\val:2^\G\to\R$, we can define the \emph{concave closure} $\bar\val:[0,1]^\G\to \R$ as 
\begin{equation}\label{eq:bar-val}
\bar\val(x):=\inf_{p\in \R^\G,\alpha\in \R}\left\{\langle p,x\rangle +\alpha: p(S) +\alpha\ge \val(S)\quad \forall S\subseteq \G\right\}\, ,
\end{equation}
see e.g. \cite[Section 3.4]{Murotabook}. As the infimum of linear functions, $\bar\val$ is always concave. Note that it provides the concave upper envelope of the function $\val$ defined on the discrete set $\{0,1\}^\G$, meaning that $\bar\val\le f$ for every concave function $f:\R_+^\G\to \R$ such that $\val(S)\le f(\chi_S)$ for all $S\subseteq \G$.

We leave it to the reader as an exercise to verify that for an additive valuation $\val(S)=\sum_{j\in S} \val_j$, the concave closure is the linear function $\bar\val(x)=\langle \val,x\rangle$.

Whereas the extension $\bar\val$ can be defined and is concave for every valuation function $\val$, evaluating $\bar\val(x)$ can be a hard problem. For example, in the case of submodular valuations, deciding whether $p(S)+\alpha\ge \val(S)$ holds for all $S\subseteq \G$ amounts to submodular maximization and is thus NP-hard. Computing $\bar\val(x)$ amounts to minimization over a polyhedron $P$ where separation is NP-hard; by the polynomial equivalence of optimization and separation \cite{gls}, it follows that evaluating $\bar\val(x)$ is NP-hard for submodular functions (see also \cite[Lemma 6.15]{submodularExtension}).

Apart from computational hardness, another problem is that $\bar\val(\chi_S)>\val(S)$ may be possible for $S\subseteq \G$. If $\bar\val(\chi_S)=\val(S)$ for all subsets $S\subseteq \G$, then we say that $\bar\val$ is the \emph{concave extension} of $\val$, and that $\val$ is \emph{concave extensible}. 

Theorem 6.43 in \cite{Murotabook} asserts that all M\nat-concave functions are concave extensible, and the converse is also essentially true. This underlines the importance of gross substitutes/M\nat-concave valuations for our approach: this is the subclass of valuations where we can naturally use convex relaxation techniques.
We also note that for M\nat-concave functions, the concave extension can be evaluated in polynomial time. This is since, in contrast with general submodular functions, M\nat-concave functions can be efficiently maximized with a simple greedy algorithm.

\paragraph{The concave extension of Rado valuations} For the case of Rado valuations, we now give an explicit description of the concave extension by a linear program.
This representation of the concave extension is at the core of the arguments in Section~\ref{section:sparsifying}, 
where we argue about the existence of a sparse optimal solution of a particular convex program. 

\begin{theorem}\label{thm:integral}
Consider a Rado valuation $v:2^\G\to \R$ given by a bipartite graph $(\G, \T; \E)$ with costs on the edges $\co : \E \to \R$,
and a matroid $ \M = (\T, \I)$ with a rank function $\f=\f_\M$ as in Definition~\ref{def:Rado}. 
For $x\in [0,1]^{\G}$, let us define
\begin{equation}\label{prog:utility}
\begin{aligned}
   \nu(x):= &\quad \max \quad \sum_{(j,k)\in \E}    \co_{jk} z_{jk}   \\
   &\begin{aligned}
    \text{s.t.: }   && \sum_{k\in \T } z_{jk} &\le x_j     &&\quad \forall j \in \G \\
                    &&\sum_{j \in \G, k \in T} z_{jk} &\le \f(T)        &&\quad \forall T \subseteq \T \\
                  &&z &\ge 0  \,.      &&
    \end{aligned} 
\end{aligned}
\end{equation}
Then, $\nu=\bar\val$ is the concave extension of $\val$, and satisfies Assumption~\ref{monotonicity}.
\end{theorem}
\begin{proof}
The function $\nu$ is clearly continuous and $\nu(\mathbf{0})=0$, thus, it is a valuation function. 
We postpone the proof that $\nu$ is the concave closure, i.e. $\nu=\bar\val$ to Lemma~\ref{lemma:concaveClosure}.
Let us now show that $\nu$ is a concave extension, namely $\nu(\chi_S)=\val(S)$ for every $S\subseteq \G$. 
First, note that whenever $M'$ is a feasible matching in the definition of $\val(S)$, $\chi_{M'}$ is a feasible solution to \eqref{prog:utility} defining $\nu(\chi_S)$.
 The left hand side of the program defining $\nu(\chi_S)$ is integral, and the feasible region of 
\eqref{prog:utility} is a linear maximization problem over the intersection of two integral submodular polytopes on $\E$. 
Using the total-dual integrality of polymatroid intersection, see
\cite[Theorem 46.1]{schrijver2003combinatorial}, the existence of an integer optimal solution $z\in \Z^E$ is guaranteed. Noting that $\f(\{v\})\le 1$ for every $v\in \T$,  it follows that $z=\chi_M$ for a matching $M$, and $\delta_{\T}(M)$ is independent in $\M$. We conclude that $\nu(\chi_S)=\val(S)$.

Let us now turn to Assumption~\ref{monotonicity}. Monotonicity is immediate. Concavity is implied by Lemma~\ref{lemma:concaveClosure}, but let us also give a simple direct proof.
 Let $x,y\in \R_+$, $\lambda\in [0,1]$, and let $z$ and $z'$ be the optimal solutions in the definition of $\nu(x)$ and $\nu(y)$. Then, it is immediate that $\lambda z+(1-\lambda)z'$ is a feasible solution for the program defining $\nu(\lambda x+(1-\lambda)y)$, showing that $\nu(x)+\nu(y)\le \nu(\lambda x+(1-\lambda)y)$.

For subadditivity, if $z$ is the optimal solution in the program defining $\nu(x+y)$ for some $x,y\in [0,1]^\G$, then we can easily decompose $z=z'+z''$ such that $z'$ is feasible to the program defining $\nu(x)$ and $z''$ if feasible for $y$. Thus, $\nu(x+y)\le \nu(x)+\nu(y)$ follows.
\end{proof}
In the light of this theorem, in the rest of the paper we will denote by $v:[0,1]^\G\to \R$ the continuous Rado valuation defined in \eqref{prog:utility}.

\subsection{Simple upper bounds} We will often use the following simple bounds.
\label{subsection:upperBoundsPreliminaries}
\begin{lemma}\label{lem:productBound} 
Let $n, c \in \N$, $S \subseteq [n]$,  and $1\le w_1, \dots, w_n\le \g-1$.
For $i\in S$ let $\s_i \in \R_+$ such that $\sum_{i \in S} \s_i \le c\cdot n$.
Then
 $$\left( \prod_{i \in S} \s_i^{w_i} \right)^{1/\sum_{i=1}^n w_i} 
\le c\cdot \g \, .$$
\end{lemma}
\begin{proof}
By the (weighted) arithmetic-geometric we have:
\begin{eqnarray*}
\left( \prod_{i \in S} \s_i^{w_i} \right)^{1/\sum_{i=1}^n w_i} 
  &=& \prod_{i \in S} \s_i^{\frac{w_i}{\sum_{i=1}^n w_i}} \cdot \prod_{i \in [n]\setminus S} 1^{\frac{w_i}{\sum_{i=1}^n w_i}} \\
  &\le& \sum_{i \in S} \frac{w_i \s_i}{\sum_{i=1}^n w_i}  + \sum_{i \in [n]\setminus S} \frac{w_i }{\sum_{i=1}^n w_i} 
  \le (\g-1) \frac{\sum_{i \in S} k_i }{\sum_{i=1}^n w_i} + 1 \le c\cdot \g\, . \qedhere
\end{eqnarray*}
\end{proof}

\begin{lemma}\label{lemma:symmetricProductBound} 
Let $n, c\in \N$, $S \subseteq [n]$.
For $i\in S$ let $\s_i \in \R_+$ such that $\sum_{i \in S} \s_i \le c \cdot n$.
Then $$\left( \prod_{i \in S} \s_i \right)^{1/n} \le c \cdot e^{1/e}\,.$$
\end{lemma}
\begin{proof}
We present the proof for $c=1$, the general cases easily reduces to $c=1$ by scaling.   
Without loss of generality, we assume that $\s_i \ge 1$ for $i \in S$. 
For fixed size of $S$ ($k = |S|$), the product $\prod_{i\in S} \s_i$ is maximized when all $\s_i$ are the same. 
Hence, $\left( \prod_{i \in S} \s_i \right)^{1/n} \le \left( \frac{n}{k} \right)^{k/n}$.
Let $\xi = \frac{n}{k}$ then $\left( \frac{n}{k} \right)^{k/n} = \xi^{1/\xi}$.
By the first order conditions, the value $ \xi^{1/\xi}$ achieves the maximum for $\xi = e$.
Hence, $\left( \prod_{i \in S} \s_i \right)^{1/n} \le e^{1/e}$.
\end{proof}

\begin{sloppypar}
Using the similar approach as in the proof of Lemma~\ref{lemma:symmetricProductBound},
one can replace the bound $\left( \prod_{i \in S} \s_i^{w_i} \right)^{1/\sum_{i=1}^n w_i} 
\le c\cdot \g$ in Lemma~\ref{lem:productBound} by the bound $\left( \prod_{i \in S} \s_i^{w_i} \right)^{1/\sum_{i=1}^n w_i} 
\le c\cdot O\left(\frac{\g}{\log(\g)}\right)$. 
This proof is deferred to Section~\ref{section:productBounds}, 
and the exact bound given there is always stronger than in Lemma~\ref{lem:productBound}.  
Moreover, trivially we have $\left( \prod_{i \in S} \s_i^{w_i} \right)^{1/\sum_{i=1}^n w_i} \le c\cdot n$.
Nevertheless, we will use only use Lemma~\ref{lem:productBound} for the asymmetric, and Lemma~\ref{lemma:symmetricProductBound} for the symmetric version of the problem in rest of the paper.
\end{sloppypar}

\section{Overview of the approach}
\label{sec:appraoch}
Let $\val_i$ be a continuous valuation function and $w_i>0$ be the weight for each $i\in \A$. Given a fractional allocation $x=(x_1,\ldots,x_n)\in \R_+^{\A\times \G}$, we let 
\[
\nsw(x):=\left( \prod_{i\in \A} \val_i(x_i)^{w_i} \right)^{1/\sum_i w_i}\, .
\]
Then, the asymmetric Nash social welfare program is captured by the following integer program.
\begin{tagequation}[NSW-IP]
\label{prog:NSW}
\max \nsw(x)\quad \textrm{s.t.} \sum_{i\in \A} x_{ij} \le 1
\  \forall j \in \G, x\in \{0,1\}^{E}\, .
\end{tagequation}
Let $\opt$ denote the optimum value.
The natural relaxation is
\eqref{prog:NSW} is
\begin{equation}\label{prog:vanilla-relax}
\max\nsw(x)\quad \textrm{s.t.} \sum_{i\in \A} x_{ij} \le 1
\  \forall j \in \G, x\ge 0\, .
\end{equation}
The objective is log-concave assuming the $\val_i$'s are concave functions. 
However, \citet[Lemma 3.1]{cole2015approximating} showed that this relaxation has 
unbounded integrality gap already for additive valuations.

We propose a mixed integer programming relaxation instead of~\eqref{prog:vanilla-relax}.
Consider a set of items $\Hs\subseteq \G$. Our mixed relaxation requires 
the items in $\Hs$ to be allocated integrally and the rest can be allocated fractionally.
\begin{tagequation}[Mixed relaxation]
\label{prog:relaxation}
\begin{aligned}
  && \max ~~ &\nsw(x)\\
   \text{s.t.: } \quad && \sum_{i\in \A} x_{ij} &\le 1        && \forall j \in \G \\
   && x_{ij} &\in \{ 0, 1 \}        && \forall j \in \Hs, \forall i \in \A\\
   &&x&\ge 0 \,.
\end{aligned}
\end{tagequation}
This clearly gives a relaxation of \eqref{prog:NSW}: $\opt_{\Hs}\ge \opt$ where $\opt_{\Hs}$ is optimal value of~\eqref{prog:relaxation} for any set of items $\Hs$.
Theorem~\ref{thm:main} is shown by constructing an integer allocation $x\in \{0,1\}^{\A\times \G}$ and an item set $\Hs$ such that $\nsw(x)\ge \opt_{\Hs}/(256\g^3)$. This is proved in five phases:
\begin{enumerate}[align = left, labelindent=\parindent, leftmargin=*, label={\bf Phase \Roman*}]
\item\label{phase1} Find an appropriate item set $\Hs$.
\item\label{phase2} Approximate~\eqref{prog:relaxation} by another integer program~\eqref{prog:decomposedUtility}.
\item\label{phase3} Find an approximate mixed integer solution to~\eqref{prog:decomposedUtility}.
\item\label{phase4} Find a \emph{sparse} approximate mixed integer solution to \eqref{prog:decomposedUtility}.
\item\label{phase5} Round the mixed integer solution to an integer solution.
\end{enumerate}
We note that phases are not necessarily algorithmic phases but also conceptional reductions of the problem. 
Regardless, we call it a phase for the simplicity of the presentation. 
We now give an overview of all the phases; most proofs are deferred to later sections.
\subsection{Phase I: Finding the item set $\Hs$}
We solve a maximum weight matching problem that achieves the highest Nash social welfare value under the restriction that each agent may only receive a single item.
This can be achieved by assigning an edge weight $\omega_{ij}=w_i\log (\val_{ij})$ for every $i\in \A$, $j\in \G$, and solving the maximum weight assignment problem in the complete bipartite graph between $\A$ and $\G$; we recall the notation $\val_{ij}=\val_i(\{j\})$. 
We let $\tau:\A\to \G$ denote the optimal matching represented as a mapping, i.e. $\tau(i)$ is the item matched to agent $i\in \A$. We define 
 $\Hs$ as the set of items assigned by $\tau$, i.e., $\Hs := \tau(\A)$.
We will refer to this set $\Hs$ as the \emph{set of most preferred items}.\footnote{
  Interestingly, in case of symmetric agents endowed with additive valuations the 
  set $\Hs$ contains all items with price at least one in any spending restricted equilibrium as in \cite{cole2015approximating};  
  see Section~\ref{section:connectionToSR}.
} 

The existence of $\tau$ with finite weight proves that the instance is feasible, i.e., 
there is a way of allocating one item to each agent such that agent values the assigned item positively. On the other hand, if no finite weight matching exists, 
the optimum value to  \eqref{prog:NSW}  is 0. 
 Henceforth, we assume without loss of generality that the optimal NSW is non-zero.

\subsection{Phase II: Reduction to the mixed matching relaxation}
We approximate~\eqref{prog:relaxation} by a second mixed integer program. 
We use variables $\y\in \R^{\A\times (\G\setminus \Hs)}$ representing the fractional allocations of the items in $\G\setminus \Hs$. Even though the valuation functions $\val_i$ are defined on $\R^\G$, we use $\val_i(\y_i)$ to denote $\val_i(x_i)$, where $x_i$ is obtained from $\y_i$ by setting $x_{ij} = 0$ for $j \in \Hs$ and $x_{ij} = y_{ij}$ for $j \in \G \setminus \Hs$.
\begin{tagequation}[Mixed+matching]
\label{prog:decomposedUtility}
\begin{aligned}
   &\max \quad \left( \prod_{i\in \A} \left( \val_i (\y_i) + \val_{i \sigma(i)} \right)^{w_i}  \right)^{1/\sum_i w_i}\\
   &\begin{aligned}
   \text{s.t.: }\quad && \sum_{i\in \A} \y_{ij} &\le 1      && \forall j \in \G \setminus \Hs\\
   && \y_{ij} &\ge 0      && \forall j \in \G \setminus \Hs, \forall i \in \A \\
   && \sigma: \A \to \Hs & \text{ is a matching.}      && 
   \end{aligned}
\end{aligned}
\end{tagequation}
We will refer to this program as the \emph{mixed matching relaxation}.
The program~\eqref{prog:decomposedUtility} differs from \eqref{prog:relaxation} in two respects.
Firstly, the objective differs from $\nsw(x)$: for each agent,  we evaluate the utility of each agent separately on $\Hs$ and $\G \setminus \Hs$.
Secondly, and more importantly, we require that the items in 
$\Hs$ are allocated to the agents by a matching. Unlike \eqref{prog:relaxation}, this will not be a relaxation of \eqref{prog:NSW}: the optimal integer solution may allocate multiple items in $\Hs$ to the same agent.
 We show that the effect of both these changes is limited.

Let $(\y, \sigma)$ be a feasible solution to~\eqref{prog:decomposedUtility}. 
We define $\overline \nsw(\y, \sigma)$  as the objective function value in \eqref{prog:decomposedUtility}, and let ${\overline\opt}_{\Hs}$ denote the optimum value.
Let us define $\nsw(\y,\sigma)$ as the Nash social welfare of the same allocation. Namely, 
$\nsw(\y,\sigma)=\nsw(x)$, 
where $x_{ij}=\y_{ij}$ if $j\in \G\setminus \Hs$, and for $j \in \Hs$  we have $x_{ij}=1$ if $j=\sigma(i)$, and $x_{ij}=0$ otherwise. The next lemma is an easy consequence of concavity and subadditivity.

\begin{lemma}\label{lem:obj-comp}
For any
feasible solution $(\y, \sigma)$  to~\eqref{prog:decomposedUtility}, we have
\[
\overline \nsw(\y, \sigma)\ge \nsw(\y,\sigma)\ge \frac{1}{2}\overline \nsw(\y, \sigma)\, .
\]
\end{lemma}
Using this lemma, as well as Lemma~\ref{lem:productBound}, we can relate the optimum values and approximate solutions of \eqref{prog:relaxation} and \eqref{prog:decomposedUtility}.
\begin{restatable}{theorem}{firstReduction}\label{theorem:reduction}
Let $\Hs\subseteq \G$ with $|\Hs|\le |\A|$. 
For the optimum values $\opt_{\Hs}$ to \eqref{prog:relaxation} 
and $\overline \opt_{\Hs}$ to \eqref{prog:decomposedUtility}, we have
\[
\overline \opt_{\Hs}\ge \frac{1}{\g} \opt_{\Hs}.
\]
Let $(\y,\sigma)$ be an $\alpha$-approximate optimal solution to \eqref{prog:decomposedUtility}, 
that is, $\overline\nsw(\y,\sigma)\ge \frac{1}{\alpha} \overline\opt_{\Hs}$.
Then, $\nsw(\y,\sigma)\ge \frac{1}{{2\alpha\g}}\opt_{\Hs}$. 
If the valuation functions $\val_i$ are additive, then the stronger bound $\nsw(\y,\sigma)\ge  \frac{1}{\alpha\g} \opt_{\Hs}$ applies.
\end{restatable}
\begin{proof}
\begin{sloppypar}
We first show that $\overline \opt_{\Hs}\ge \frac{1}{\g} \opt_{\Hs}$.
Let $x$ be an optimal solution to~\eqref{prog:relaxation}.
For each agent $i$, let $K_i$ be the set of items agent $i$ 
receives from $\Hs$ under $x$; and let $\y$ be the restriction of $x$ on $\G \setminus \Hs$
defined as $\y_{ij} = x_{ij}$ for $j\in \G \setminus \Hs$ and $\y_{ij} = 0$ otherwise.
Let $k_i := |K_i|$.
Denote with $S$ the set of agents that receive at least one items from $\Hs$,
i.e., $S = \{i \in \A: k_i \ge 1\}$. 
For each agent $i\in S$ let $\sigma(i) = \max_{j \in K_i}\{ \val_{ij}\}$,
and define $\sigma(i) = \emptyset$ for $i \in \A \setminus S$. 
Then, $(\y, \sigma)$ is a feasible solution of~\eqref{prog:decomposedUtility}. 
In other words, $(\y, \sigma)$ is obtained from $x$ once each agent $i \in S$ discards all items from $K_i$ 
except the most valuable one.
By monotonicity and subadditivity, for all $i\in S$, we have
\begin{equation*}
\val_i (x_i) \le \val_i (\y) +  \sum_{j \in K_i} \val_{ij} \le k_i \cdot (\val_i (\y) + \val_{i \sigma(i)})\,.
\end{equation*}
\end{sloppypar}
Therefore,
$$
\frac{\opt_{\Hs}}{\overline \opt_{\Hs}} \le 
 \frac{\nsw(x)}{\overline{\nsw}(\y, \sigma)} 
= \left(\prod_{i \in S}  \frac{\val_i(x_i)^{w_i} }{(\val_i (\y) + \val_{i \sigma(i)})^{w_i}} \right)^{1/\sum_i w_i} 
\le \left(\prod_{i \in S} k_i^{w_i} \right)^{1/\sum_i w_i} \,.$$
Moreover, $\sum_{i \in S} k_i \le |\Hs| \le |\A| = n$.
Then, the bound follows by Lemma~\ref{lem:productBound}.
The second part of the theorem follows by Lemma~\ref{lem:obj-comp}.
\end{proof}

\subsection{Phase III: Approximating the mixed matching relaxation}
Our next goal is to find a $2$-approximation solution to \eqref{prog:decomposedUtility}; we do not know whether this problem is polynomial-time solvable.
By Theorem~\ref{theorem:reduction}, this yields a $(4\g)$-approximation to \eqref{prog:relaxation}. 

Let us first remove all items in $\Hs$. Some agents may only value positively the items $\Hs$. 
We let $\A'$ the subset of agents who have positive values for the items $\G \setminus \Hs$, that is, 
$\A' := \{ i \in \A : \val_i(\G \setminus \Hs) > 0 \}$.
Consider the ``na\"\i ve'' relaxation \eqref{prog:vanilla-relax} on the instance restricted to $\A'$ and $\G\setminus \Hs$, and taking the logarithm of the objective
\begin{tagequation}[EG]
\label{prog:EGprod}
\begin{aligned}
  & \max \quad \sum_{i\in \A} w_i \log(\val_i(y_i))\\
  &\begin{aligned}
    \text{s.t.: } \quad && \sum_{i \in \A'} \y_{ij} &\le 1    &&\quad \forall j \in \G \setminus \Hs\\
            && \y &\ge 0.
   \end{aligned}
\end{aligned}
\end{tagequation}
This is the classical Eisenberg--Gale convex program that computes an equilibrium in Fisher markets with divisible items for homogeneous concave valuation functions~\cite{Eisenberg1961}.
Given an optimal solution $\y^*\in \R_+^{\A'\times (\G\setminus \Hs)}$ of~\eqref{prog:EGprod} 
we can find an approximate solution to~\eqref{prog:decomposedUtility}.

\begin{theorem}\label{theorem:approx}
Let $\Hs\subseteq \G$ with $|\Hs|\le |\A|$. 
Let $\pi^*$ be maximum weight assignment in the complete bipartite graph between $\A$ and $\Hs$, 
with edge weights $\omega_{ij}=w_i \log \left( \val_i(\y^*_i) + \val_{ij} \right)$ for $i\in \A$, $j\in \Hs$. 
Then, $\overline \nsw(\y^*, \pi^*) \ge \frac{1}{2} \overline \opt_{\Hs}$.
\end{theorem}

Theorem~\ref{theorem:approx} is an immediate consequence of the following lemma.

\begin{restatable}{lemma}{approximatingMIP}
\label{lem:approximatingMIP}
Let $\Hs\subseteq \G$ with $|\Hs|\le |\A|$. 
Let $\alpha>0$ and $\y^*$ be an optimal and $\y$ a feasible solution of~\eqref{prog:EGprod} such that 
$\val_i(\y_i) \ge \frac{1}{\alpha} \val_i(\y^*_i)$ for all $i \in \A'$.
Let $\pi$ be maximum weight assignment in the bipartite graph with colour classes $\A$ and $\Hs$, 
and edge weights $\omega_{ij}=w_i \log \left( \val_i(\y_i) + \val_{ij} \right)$ for $i\in \A$, $j\in \Hs$. 
Then,
\[
 \overline \nsw(\y, \pi) \ge \frac{1}{2\alpha} \overline \opt_{\Hs}\, .
 \]
\end{restatable}

Since valuations $\val_i$ are concave,~\eqref{prog:EGprod} is a convex program. 
For any $\varepsilon>0$, we can find an $(1-\varepsilon)$-approximate solution  in polynomial-time, 
where the running time depends on $\log(1/\varepsilon)$.
It turns out that approximation of the objective function might not be enough. 
In Lemma~\ref{lem:approximatingMIP} we require an agent-wise approximate solution:
each agent gets at least a constant fraction of her value in the optimum. 
It is not clear if finding such agent-wise approximation is possible in polynomial time for general concave valuations $\val_i$,
but as we will see in the next section we can find an exact optimal solution for Rado valuations.

The proof of Lemma~\ref{lem:approximatingMIP} is deferred to Section~\ref{sec:approximatingMatchingRelaxation}. 
It does not depend on the choice of $\Hs$ but only requires $|\Hs|\le |\A|$.

\subsection{Phase IV: A sparse approximate solution for the mixed matching relaxation}
\begin{sloppypar}
In this section we exploit the properties of Rado valuations.
Assuming the agents have Rado valuation functions, 
we can find an approximate solution of~\eqref{prog:decomposedUtility} with a strong sparsity property. 
Even though the approximation ratio is weaker then given in Theorem~\ref{theorem:approx}, 
sparsity will be essential for the rounding in \ref{phase5}.
\end{sloppypar}

\begin{restatable}{theorem}{approx-sparse}\label{theorem:approx-sparse}
Suppose the functions $\val_i$ are Rado valuations. Let $\Hs\subseteq \G$ with $|\Hs|\le |\A|$.
We can find a feasible solution $(\y,\pi)$ to \eqref{prog:decomposedUtility} such that 
\begin{enumerate}[label=(\roman*)]
\item $\overline \nsw(\y, \pi) \ge \frac{1}{4}\overline \opt_{\Hs}$,
\item $\supp(\y)\le 2|\A|+|\F^+|$ where $\F^+ = \{j \in \G\setminus \Hs: \sum_{i \in \A'} y_{ij} > 0\}$, that is, $\F^+$ is the set of allocated items in $y$.
\end{enumerate}
Moreover, for additive valuation functions, we can strengthen {\em (i)} to $\nsw(\y, \sigma)\ge \frac{1}{2}\opt_{\Hs}$ and {\em (ii)} to $\supp(\y)\le |\A|+|\F^+|$.
\end{restatable}

Let us start with the special case of additive valuations. 
In this case, an exact solution $\y^*$ to the Eisenberg--Gale convex program \eqref{prog:EGprod} 
can be found in
strongly polynomial time~\cite{orlin2010improved,vegh2016strongly}. 

\begin{theorem}\label{thm:linearEG}
Assuming the valuations $\val_i$ are additive, we can find an optimal solution $\y^*$ of~\eqref{prog:EGprod} in strongly polynomial time such that the support $\supp(\y^*)$ is a forest.
\end{theorem}
The claim on the support follows easily by showing that any cycles in $\supp(\y^*)$ can be eliminated, see e.g.,~\cite{cole2015approximating,Duan2016,orlin2010improved}. 
Consequently, 
$|\supp(\y^*)| \le |\A'| + |\F^+| -1$. 
Together with Lemma~\ref{lem:approximatingMIP}, this proves the statement in Theorem~\ref{theorem:approx-sparse} for additive valuations.

For Rado valuations, we first prove that an optimal solution of~\eqref{prog:EGprod} can be found in polynomial time, 
see Section~\ref{subsection:optimalEG}. We first show that this is a rational convex program, and use the variant of the ellipsoid method for rational polyhedra \cite{gls}.
\begin{restatable}{lemma}{optimalSolution}
\label{lemma:optimalSolution}
Suppose that for each agent $i\in \A$,  $\val_i$ is a Rado valuation 
given by a bipartite graph $(\G, \T_i; \E_i)$, integer costs 
$\co_i : \E_i \to \Z$ and a matroid 
$\M_i = (\T_i, \I_i)$ as in Definition~\ref{def:Rado}. Let $T=\max_{i\in \A}|\T_i|$, and $C=\max_{i\in \A}\|\co_i\|_\infty$. Let the weights $w_i>0$ be rational numbers given as quotients of two integers at most $U$. 
Assume  the matroids $\M_i$ are given by rank oracles.
Then, \eqref{prog:EGprod} has a rational solution with 
$\mathrm{poly}(|\A|,|\G|,T,\log C, \log U)$ bit-complexity, and  such a solution can be found 
in $\mathrm{poly}(|\A|,|\G|,T,\log C, \log U)$  arithmetic operations and calls to the matroid rank oracles.
\end{restatable}

Our next lemma shows that any feasible solution to~\eqref{prog:EGprod} 
can be sparsified by losing at most the half of the value for each agent, see Section~\ref{subsection:sparseOptimalSolution}. This is achieved in two steps, using the sparsity of basic feasible solutions to linear programs. Half of the valuation may be lost in the second step, where for the fractionally allocated items we aim to remove one of the fractional edges. The set to be deleted is identified by  writing an auxiliary linear program.
\begin{restatable}{lemma}{sparseSolution}
\label{lem:sparseSolution}
Suppose the functions $\val_i$ are Rado valuations, and let $\hat \y$ be a feasible solution to 
\eqref{prog:EGprod}. Then, in polynomial time we can find 
a feasible solution $\y$ such that
\begin{enumerate}[label=(\roman*)]
  \item $\val_i (\y) \ge \frac{1}{2} \val_i(\hat \y)$,
  \item $|\supp(\y)| \le 2|\A'| + |\F^+|$ where $\F^+ := \F^+(\y) = \{j \in \G\setminus \Hs: \sum_{i \in \A'} \y_{ij} > 0\}$.
\end{enumerate}
\end{restatable}

By combining Lemmas~\ref{lem:approximatingMIP},~\ref{lemma:optimalSolution},~\ref{lem:sparseSolution},
 we obtain Theorem~\ref{theorem:approx-sparse} for Rado valuations.

\subsection{Phase V: Rounding the mixed integer solution}
For this phase of the algorithm, we require a sparse approximate solution as in Theorem~\ref{theorem:approx-sparse}, 
and exploit the choice of $\Hs$ as the set of most preferred items in \ref{phase1}. We start with a mixed integer solution $(\y,\pi)$ as in Theorem~\ref{theorem:approx-sparse}. 
By a \emph{reduction} of $(\y,\pi)$ we mean a mixed integer solution $(\ry,\pi)$ obtained as follows.
For each $j\in  \F^+$, we pick an arbitrary agent $\kappa(j)\in \A$ such that $\y_{\kappa(j)j}>0$. 
We set $\ry_{\kappa(j)j}=\y_{\kappa(j)j}$, and set $\ry_{ij}=0$ if $i\neq\kappa(j)$. 
By the bound on $\supp(\y)$, this amounts to setting $\le 2|\A|$ values $\y_{ij}$ to 0. 
The proof of the next lemma is given in Section~\ref{section:rounding}.

\begin{restatable}{lemma}{newMatchingMain}
\label{lem:newMatchingMain}
Let $\Hs$ be the set of most preferred items, and let  $(\y,\pi)$ be a solution to \eqref{prog:decomposedUtility} as in Theorem~\ref{theorem:approx-sparse}. Let $(\ry,\pi)$ be a reduction 
of $(\y,\pi)$.
Then in polynomial-time we can find a matching $\rho: \A\to\Hs$ such that 
$$ \overline \nsw(\ry, \rho) \ge \displaystyle \frac{1}{32\g^2}\overline \nsw(\y, \pi)\,.$$
Further, if the valuations are linear, then we can find a matching $\rho: \A\to\Hs$ such that 
$ \overline \nsw(\ry, \rho) \ge \frac{1}{8} \overline \nsw(\y, \pi)$.
\end{restatable}

Such a matching $\rho$ can be found by combining the matching $\pi$ in the solution $(\y,\pi)$, and the initial matching $\tau$ from \ref{phase1} that delivers the highest NSW value such that every agent may receive only one item. We swap from $\pi$ to $\tau$ on certain alternating cycles. 

\medskip

We are ready to prove the main results. 

\main*
\begin{proof}
From Theorem~\ref{theorem:approx-sparse} and Lemma~\ref{lem:newMatchingMain}, we can obtain a solution 
an $(128\g^2)$-approximate solution $(\ry,\rho)$ to \eqref{prog:decomposedUtility} 
such that for each item  $\F^+$ there is exactly one incident edge in $\supp(\ry)$.
We can obtain a 0--1 valued solution $x$ to \eqref{prog:NSW} by assigning each item in $\Hs$ 
according to $\rho$ and each item $j\in \F^+$ to the unique agent $i$ with $\ry_{ij}>0$. 
Clearly, $\nsw(x)\ge \nsw(\ry,\rho)$. 
We obtain $\nsw(x)\ge \opt_{\Hs}/({256\g^3})\ge \opt/({256\g^3})$ using Theorem~\ref{theorem:reduction}.
For additive valuations, we use the stronger bounds in the same results.
\end{proof}

\mainSym*
\begin{proof}
The proof follows exactly as the proof of Theorem~\ref{thm:main} once we replace $\gamma$ by $e^{1/e}$.
Such a change is justified as in the symmetric case we can use Lemma~\ref{lemma:symmetricProductBound} instead of the bound given by 
Lemma~\ref{lem:productBound}.
\end{proof}

\section{Phase III: Approximating the mixed matching relaxation}
\label{sec:approximatingMatchingRelaxation}{}
\ref{phase3} presents a general way of obtaining a $2$-approximation to~\eqref{prog:decomposedUtility}. 
By Theorem~\ref{theorem:reduction}, this gives a $(4\g)$-approximation to~\eqref{prog:relaxation}, a mixed integer relaxation of the ANSW problem.
Recall that~\eqref{prog:decomposedUtility} is the following mixed integer program 
\begin{tagequation}[Mixed+matching]
\begin{aligned}
   &\max \quad \left( \prod_{i\in \A} \left( \val_i (\y_i) + \val_{i \sigma(i)} \right)^{w_i}  \right)^{1/\sum_i w_i}\\
   &\begin{aligned}
   \text{s.t.: }\quad && \sum_{i\in \A} \y_{ij} &\le 1      && \forall j \in \G \setminus \Hs\\
   && \y_{ij} &\ge 0      && \forall j \in \G \setminus \Hs, \forall i \in \A \\
   && \sigma: \A \to \Hs & \text{ is a matching.}      && 
   \end{aligned}
\end{aligned}
\end{tagequation}

In the above problem, we need to allocate items $\G$ to the agents in $\A$ in order to maximize an objective function that 
is an approximation of the NSW. Items in $\G \setminus \Hs$ can be allocated fractionally to the agents without any constraints. 
The items in $\Hs$ have to be allocated integrally via an assignment, thereby allocating exactly one item from $\Hs$ to each agent $\A$.

While the exact computational complexity of~\eqref{prog:decomposedUtility} remains unresolved, 
we show that we can $2$-approximate it.

Denote $\F = \G \setminus \Hs$. 
Let $\A'$ be the subset of agents that have positive value for the items in $\G \setminus \Hs$, $\A' := \{ i \in \A : \val_i(\G \setminus \Hs) > 0 \}$, as some agents may only have positive value for the items in $\Hs$.
Restricting~\eqref{prog:decomposedUtility} to the items $\F$ and agents $\A'$ and taking the objective  yields an instance of \eqref{prog:EGprod}:
\begin{equation*}
\begin{aligned}
  & \max \quad  \sum_{i\in \A'} w_i \log \val_i (\y_{i})  \\
  &\begin{aligned}
    \text{s.t.: } \quad && \sum_{i \in \A'} \y_{ij} &\le 1    &&\quad \forall j \in \F\\
            && \y_{ij} &\ge 0   &&\quad \forall j \in \F, \forall i \in \A'.
   \end{aligned}
\end{aligned}
\end{equation*} 
The above is a convex program whenever the valuations $\val_i(.)$ are concave, 
and we can solve it to an arbitrary precision in polynomial time if we have access to a supergradient oracle to the objective function. 

On the other hand, suppose that the variables $y$ are fixed in~\eqref{prog:decomposedUtility}.
Under the fixed $y$, we can find an optimal assignment $\sigma$.
Namely, an optimal assignment is exactly a maximum weight assignment in the bipartite graph 
$(\A, \Hs; E)$ where the weight of an edge $ij$ for $i\in \A$, $j\in \Hs$ is $\omega_{ij} := w_i \log( \val_i(\y_i) + \val_{ij})$.

Informally,~\eqref{prog:decomposedUtility} is a combination of two tractable problems. 
We show that an optimal solution $y^*$ to the restriction of the problem to $\F$ and $\A'$, 
and an optimal assignment with respect to the fixed $y^*$ gives a $2$-approximation for~\eqref{prog:decomposedUtility}.

In Section~\ref{subsection:propertiesOfEG} we discuss the restriction of the problem to $\F$ and $\A'$ and give a technical lemma. The main result of the section is presented in Section~\ref{subsection:part2phaseIII}.

\subsection{Properties of Eisenberg--Gale program}
\label{subsection:propertiesOfEG}
Let us now consider the Eisenberg--Gale program \eqref{prog:EGprod}.
An optimal solution $y^*$ and the optimal Lagrange multipliers $p_j$ for $j\in \F$ can be interpreted as the so-called 
\emph{Gale equilibrium} in the market with divisible items $\F$, agents $\A'$,
 and where agent $i$ has valuation $\val_i$ and budget $w_i$.
In particular,  $\y^*$ represent the allocations and 
 $p_j$ for $j \in \F$, specify the prices in the market equilibrium, 
see e.g.,~\cite{garg2019auction,nesterov2018computation}.\footnote{
  In case of homogeneous valuations this can be used to find a \emph{Fisher} equilibrium,
  since Fisher and Gale equilibria coincide under homogeneous valuations~\cite{Eisenberg1959,nesterov2018computation}. 
}

Our technical lemma relates the combined difference in valuations of each agent 
in the optimal solution $y^*$ and any other allocation $y'$. 
The rest of Section~\ref{subsection:propertiesOfEG} is devoted to its proof.

\begin{lemma}\label{lemma:sumOfFractions}
Let $\y^*$ be an optimal solution to~\eqref{prog:EGprod}. 
Then for any feasible solution $\y'$ and any $\A'' \subseteq \A'$ it holds 
$$\sum_{i \in \A''} w_i \frac{\val_i(\y'_i)}{\val_i(\y^*_i)}  \le \sum_{i \in \A''} w_i + \sum_{i\in \A'} w_i\,.$$
\end{lemma}
We recall some definitions and the Karush--Kuhn--Tucker (KKT) optimality conditions in terms of subgradients; 
see~\cite[Chapter 2 and Theorem 3.27]{ruszczynski2011nonlinear}.
Given a \emph{convex} function $f: \R^{M} \to \R$, we say that $g$ is a \emph{subgradient} of $h$ at
 $y^* \in \R^{M}$ if $f(y) \ge f(y^*) +  g^{\top} (y-y^*)$ for all $y \in \R^{M}$. 
The set of all subgradients at a point $y^*$ is  
called \emph{subdifferential} and denoted as $\partial f(y^*)$.
If the function is differentiable then $\partial f(y^*) = \{\nabla f(y^*)\}$.
Consider the convex program 
\begin{equation*}\label{prog:convex}
\begin{aligned}
  &\min \quad f_0 (y)\\
  &\begin{aligned}
    \text{s.t.: } && f_j (y) &\le 0     &&\quad \forall j \in \F \\
                  && y &\ge 0 \ , &&
   \end{aligned}
\end{aligned}
\end{equation*}
where $f_j$ for $j\in \{0\} \cup \F$ is convex. 
Assume that the there exists a strict feasible point (Slater's condition).
Then,  $y^*$ is a an optimal solution with the Lagrange multipliers $p$,
if and only if the following conditions hold
\begin{itemize}
    \item $f_j(y^*) \le 0$, $p_j \ge 0$ for all $j \in \F$ (primal and dual feasibility), 
    \item $0 \in \partial f_0 (y^*) + \sum_{j \in \F} p_j \partial f_j (y^*) + \{ \mu \in \R^{M}_{-} : \mu^\top y^* = 0\}$ (stationarity), and 
    \item $p_j f_j(y^*) = 0$ (complementary slackness). 
\end{itemize}

We say that $g$ is a \emph{supergradient} of the concave function $f$ if $-g$ is a subgradient of $-f$. 
The following proposition guarantees the existence of supergradients.
\begin{proposition}
The function $f: \R^{M}_+ \to \R$ is concave if and only if $\forall y^* \in \R^M_+$ it has a non-empty superdifferential at $y^*$.
In other words, there is $g\in \R^M$
such that 
$$
f(y) \le f(y^*) + g^{\top}(y-y^*)\, .
$$
\end{proposition}

We can interpret the Lagrange multipliers in \eqref{prog:EGprod} as  prices; the next 
 claim states that no agent spends more that her budget in a Gale--equilibrium.
\begin{claim}\label{claim:priceBoundedByWeight}
Let $\y^*$ be an optimum and $p$ be the optimal Lagrange multipliers of~\eqref{prog:EGprod}.
For all $i \in \A'$ it holds $p^{\top} \y^*_i \le w_i$.
\end{claim} 

\begin{proof}
Let us apply the above KKT conditions to the concave maximization program~\eqref{prog:EGprod}. 
for each agent  $i \in \A'$
$$0 \in \partial \left( - w_i \log (\val_i (\y^*_i)) \right) + p + \{\mu_i\in \R^{\F}_{-} : \mu_i^\top \y^*_i = 0\}\,.$$
By the  composition rules for subgradients we have
$$0 \in  -\frac{w_i \partial \val_i (\y^*_i)}{\val_i(\y^*_i)} + p + \{\mu_i\in \R^{\F}_{-} : \mu_i^\top \y^*_i = 0\}\,.$$ 
Therefore, there exists a supergradient $g_i \in \partial \val_i (\y^*_i)$ such that 
$w_i g_i^{\top} = \val_i(\y^*_i) \cdot(p^{\top} + \mu^{\top}_i)$ where $\mu_i \le 0$ and $\mu_i^\top \y^*_i = 0$.

By definition of subgradient (supergradient) at $\y^*_i$,
we have that $g_i^{\top} \y^*_i \le \val_i(\y^*_i)$ for all $i \in \A'$.
It follows that $p^{\top} \y^*_i \le w_i $ for all $i\in \A'$. 
\end{proof}

\begin{proof}[Proof of Lemma~\ref{lemma:sumOfFractions}]
By the KKT conditions, for each $i\in \A'$, we have a supergradient 
$g_i \in \partial \val_i (\y^*_i)$ such that 
$
\frac{w_i g_i}{\val_i(\y^*_i)} \le p
$ holds.
By complementarity slackness, if $p_j > 0$ then  $\sum_{i\in \A'} \y^*_{ij}=1$.
Let $\overline y_{ij} = \max\{\y^*_{ij}, \y'_{ij}\}$. Then we obtain:
$$
\val_i(\y'_i) \le \val_i(\overline y_i) 
\le \val_i(\y^*_i) +  g_i^{\top} (\overline y_i - \y^*_i)  
\le \val_i(\y^*_i) +  \frac{\val_i(\y^*_i) p^{\top}}{w_i}\cdot (\overline y_i - \y^*_i) \,.
$$
The first inequality is by monotonicity, the second by the definition of the supergradient, and the third from the KKT conditions as noted above.
After rearranging we obtain
$\displaystyle \frac{w_i \val_i(\y'_i)}{\val_i(\y^*_i)} 
\le w_i + p^{\top} (\overline y_i - \y^*_i)$.
Summing the previous inequality for each agent $i \in \A''$ for a subset $\A'' \subseteq \A'$, 
and by definition of $\overline y_i$, we have
$$
\sum_{i \in \A''} \frac{w_i \val_i(\y'_i)}{\val_i(\y^*_i)} 
\le \sum_{i \in \A''}w_i + \sum_{i \in \A''} p^{\top} (\overline y_i - \y^*_i)  
\le \sum_{i \in \A''}w_i + p^{\top} \1 \,.
$$
Since $p_j > 0$ implies $\sum_{i\in \A'} \y^*_{ij}=1$ we have that 
$p^{\top} \1 =  p^{\top}\sum_{i\in \A'} \y^*_i$.
Then, by Claim~\ref{claim:priceBoundedByWeight} we have
\begin{equation*}
\displaystyle \sum_{i \in \A''} w_i\frac{\val_i(\y'_i)}{\val_i(\y^*_i)}  \le \sum_{i \in \A''} w_i  + \sum_{i\in \A'} w_i
  \qedhere
\end{equation*}
\end{proof}

\subsection{The approximation guarantee for the mixed matching relaxation}
\label{subsection:part2phaseIII}
\approximatingMIP*

\begin{proof}
\begin{sloppypar}
Let $\pi^*$ be maximum weight matching in the bipartite graph with colour classes $\A$
and $\Hs$ and with edge weights $q^*_i = w_i \log(v_i(\y^*) + v_{ij})$. 
Equivalently, $\pi^*$ is a matching maximizing 
$$\displaystyle \left( \prod_{i\in \A'} \left( \val_i (\y^*_i) + \val_{i \pi^*(i)}\right)^{w_i} \right)^{1/\sum_{i\in \A} w_i}\,.$$
We have the bounds 
\begin{equation}\label{eq:y-pi}
\overline \nsw(\y, \pi) \ge \overline \nsw(\y, \pi^*) \ge \frac{1}{\alpha} \overline\nsw (\y^*, \pi^*)\, .
\end{equation}
The first inequality is by the definition of $\pi$ as the maximum weight matching. The second inequality follows from the assumption $\val_i(\y_i) \ge \frac{1}{\alpha} \val_i(\y^*_i)$ for each $i\in \A'$. 

The rest of the proof is devoted to proving that $\overline \nsw(\y^*, \pi^*) \ge \frac{1}{2} \overline \opt_{\Hs}$; together with \eqref{eq:y-pi}, this implies the statement.
Let us introduce some notation. 
For an agent $i \in \A$, let $\ls^*_i = \val_i(\y^*_i)$ be the value agent $i$ gets from the optimal fractional bundle $\y^*$.
Then, 
$\overline \nsw (\y^*, \pi^*) = 
\left( \prod_{i \in \A'} (\ls^*_i + v_{i \pi^*(i)})^{w_i} 
\prod_{i \in A\setminus \A'} v_{i \pi^*(i)}^{w_i} \right)^{1/\sum_{i\in \A} w_i}$.

Let $(\y', \varrho)$ be an optimal solution achieving $\overline \opt_{\Hs}$.
For an agent $i \in \A$ let $\ls_i = \val_i (\y'_i)$ be the value agent $i$ gets from the fractional allocation $\y'$.
Then $\overline \opt_{\Hs} = \overline \nsw (\y', \varrho) = \left( \prod_{i \in \A} (\ls_i + \val_{i \varrho(i)})^{w_i} \right)^{1/ \sum_{i\in \A} w_i}$.
By definition of the set $\A'$, the agents in $\A \setminus \A'$ do not value the items in $\F$.
Thus, by monotonicity
\begin{equation*}\label{equation:A'}
\overline \nsw (\y', \varrho) = \left( \prod_{i \in \A'} (\ls_i + \val_{i \varrho(i)})^{w_i} 
\prod_{i \in \A \setminus \A' } \val_{i \varrho(i)}^{w_i}  \right)^{1/ \sum_{i\in \A} w_i} \,.
\end{equation*}
By the choice of $\pi^*$, we have
\begin{equation*}
   \overline \nsw(\y^*, \pi^*) \ge \overline \nsw(\y^*, \varrho) 
   = \left(\prod_{i \in \A'} (\ls^*_i + \val_{i \varrho(i)})^{w_i} \prod_{i \in \A\setminus \A'} \val_{i \varrho(i)}^{w_i} \right)^{1/\sum_{i \in \A} {w_i}} \,.
\end{equation*}
Combining the last two we have: 
$\displaystyle
    \frac{\overline \nsw(\y', \varrho)}{\overline \nsw(\y^*, \pi^*)} 
    \le \left(\prod_{i \in \A'} \left( \frac{\ls_i + \val_{i \varrho(i)}}{\ls^*_i + \val_{i \varrho(i)}} \right)^{w_i} \right)^{1/ \sum_{i\in \A} w_i} 
$.

Let $\A'' = \{i \in \A': \ls_i > \ls^*_i\}$ be the set of agents that get more value from $y'$ than $y^*$. 
Then, for $i  \in \A' \setminus \A''$ the fraction $\displaystyle \frac{\ls_i + \val_{i \varrho(i)}}{\ls^*_i + \val_{i \varrho(i)}}$ is trivially bounded by $1$. 
On the other hand, for $i \in \A''$ we have $\displaystyle  \frac{\ls_i + \val_{i \varrho(i)}}{\ls^*_i + \val_{i \varrho(i)}} \le \frac{\ls_i}{\ls^*_i}$. 
Since $\overline \opt_{\Hs} = \overline \nsw(\y', \varrho)$ it follows
\begin{equation*}
    \frac{\overline \opt_{\Hs}}{\overline \nsw(\y^*, \pi^*)} 
    \le \left(\prod_{i \in \A'} \left(\frac{\ls_i + \val_{i \varrho(i)}}{\ls^*_i + \val_{i \varrho(i)}}\right)^{w_i} \right)^{1/\sum_{i\in \A} w_i }
    \le \left(\prod_{i \in \A''} \left(\frac{\ls_i}{\ls^*_i} \right)^{w_i} \right)^{1/\sum_{i\in A} w_i } \,.
\end{equation*}
We claim that the last expression is bounded by $2$.
By Lemma~\ref{lemma:sumOfFractions} we have $\sum_{i\in \A''} w_i\frac{\ls_i}{\ls^*_i} \le \sum_{i \in \A''}w_i + \sum_{i\in \A'}  w_i$. 
Then by the weighted arithmetic-geometric mean we have 
\end{sloppypar} 
\[
\begin{aligned}
\prod_{i \in \A''}  \left( \frac{\ls_i}{\ls^*_i} \right)^{w_i/ \sum_{i\in \A} w_i } 
\le \frac{\sum_{i\in \A''} w_i \frac{\ls_i}{\ls^*_i}  +\sum_{i \in \A\setminus \A''} 1}{\sum_{i \in \A} w_i} 
\le \frac{\sum_{i \in \A''}w_i + \sum_{i\in \A'} w_i +  |\A\setminus \A''|}{\sum_{i \in \A} w_i} 
\le 2 \,.
\end{aligned}
\]
The lemma follows. 
\end{proof}

\section{Phase IV: Obtaining a sparse approximate solution}
\label{section:sparsifying}

Recall that a continuous Rado valuation is defined as an optimum of the LP \eqref{prog:utility}. For the valuation $\val_i$ of agent $i\in\A$, this is defined by
 a bipartite graph $(\G, \T_i; \E_i)$ with costs on the edges $\co_i : \E_i \to \R$,
and a matroid $ \M_i = (\T_i, \I_i)$ with a rank function $\f_i=\f_{\M_i}$.
The program~\eqref{prog:EGprod} for 
$\A'$ and $\F= \G \setminus \Hs$ can be thus written as follows.
\begin{equation*}\label{prog:explicit}
\begin{aligned}
    &\text{max} \quad \sum_{i\in \A'} w_i \log \left(\sum_{j \in \F} \sum_{k\in \T_i} \co_{ijk} z_{ijk}   \right) \\
    &\begin{aligned}
    \text{s.t.: }\quad &&\sum_{i \in \A'} \y_{ij} &\le 1          &&\quad \forall j \in \F\\
                    && \sum_{k\in \T_i } z_{ijk} &\le \y_{ij}       &&\quad \forall i\in \A' , \forall j \in \F\\
                    &&\sum_{j \in \F} \sum_{ k \in \Se} z_{ijk} &\le \f_i(\Se)           &&\quad \forall i \in \A', \forall \Se \subseteq \T_i \\
                    &&\y  \ge 0\, , \quad   z &\ge 0    \,.          &&
    \end{aligned}
\end{aligned}
\end{equation*}
Without loss of generality we can assume that the second set of constraints always holds with equality, i.e., 
$\y_{ij} = \sum_{k \in \T_i} z_{ijk}$ for $j\in \F$ and $i \in \A'$.
By  eliminating the variables $\y$, the program~\eqref{prog:EGprod} becomes:
\begin{tagequation}[EG-Rado]
\label{prog:explicitReduced}
\begin{aligned}
 &\text{max} \quad \sum_{i\in \A'} w_i \log \left(\sum_{j \in \F} \sum_{k\in \T_i} \co_{ijk} z_{ijk}   \right) \\
    &\begin{aligned}
    \text{s.t.: }\quad&& \sum_{i \in \A'} \sum_{k \in \T_i} z_{ijk} &\le 1          &&\quad \forall j \in \F\\
                    && \sum_{j \in \F}\sum_{ k \in \Se} z_{ijk}  &\le \f_i(\Se)     &&\quad \forall i \in \A', \forall \Se \subseteq \T_i \\
                    && z &\ge 0  \,,                                        && 
    \end{aligned}
\end{aligned}
\end{tagequation}
Using this formulation, we first show that the Eisenberg--Gale type convex program~\eqref{prog:EGprod} can be solved exactly in polynomial time for Rado valuations (Section~\ref{subsection:optimalEG}).
We then transform the optimal solution to a sparse approximate solution (Section~\ref{subsection:sparseOptimalSolution}).

\subsection{Solving the Eisenberg-Gale relaxation}
\label{subsection:optimalEG}

In this section, we prove the following lemma.

\optimalSolution*

As noted above,  \eqref{prog:EGprod} with Rado valuations for the set of agents $\A'$ and set of goods $\F$ is equivalent to \eqref{prog:explicitReduced}. Throughout, we assume this program is feasible, i.e. it has a solution with finite objective value. This is a mild condition only requiring the existence of at least one edge $(j,k)\in E_i$ with $c_{ijk}>0$ and $r_i(\{k\})=1$ for every $i\in \A'$.

In general, one can only expect to solve convex programs approximately: no rational solution may even exist. Vazirani \cite{Vazirani2012} defines rational convex programs where a finite optimum exists with bounded bit-complexity in the input size, where the input is described by a finite set of parameters. 
This model is not directly applicable for our program \eqref{prog:explicitReduced} as it is described with an exponential number of constraints. The bound $\mathrm{poly}(|\A|,|\G|,T,\log C, \log U)$ does not take into account the matroidal constraints; it is polynomial in the amount of information needed to describe the objective function.\footnote{We note that for exponential size linear programs, a standard way to bound the encoding size is giving bounds on facet/vertex-complexity, defined later in this section. The program \eqref{prog:explicitReduced} maximizes a concave function over a polytope that has facet complexity $O(|\A|T)$.}

We first show that the set of optimal solutions is a polytope where the vertices have polynomially bounded bit-complexity.

\begin{lemma}\label{lemma:rationalOptimum}
For an NSW problem instance with Rado valuations as in Lemma~\ref{lemma:optimalSolution}, the set of optimal solutions forms a polytope. The bit-complexity of each vertex of this polytope is bounded as $\mathrm{poly}(|\A|,|\G|,T,\log C, \log U)$.
\end{lemma}
To prove the above lemma we use the KKT conditions for \eqref{prog:explicitReduced}.
Let $p_j$'s and $\alpha_i(S)$'s denote the Lagrange multipliers corresponding to the first and second sets of the constraints, respectively. 
It holds:
\begin{enumerate}[label= $(\roman*)$]
    \item\label{egKKT1} $\forall j \in \F : p_j \ge 0$.
    \item\label{egKKT2} $\forall i \in \A', \forall S\subseteq \T_i: \alpha_i(S)\ge 0$.
    \item\label{egKKT3} $\displaystyle \forall j \in \F : 
            p_j > 0 \implies \sum_{i \in \A', k\in \T_i} z_{ijk} = 1$.
    \item\label{egKKT4} $\displaystyle \forall i \in \A', \forall S\subseteq \T_i:
                \alpha_i(S) > 0 \implies \sum_{j \in \F, k \in S} z_{ijk} = \f_i(S)$.
    \item\label{egKKT5} $\displaystyle \forall i\in \A', \forall j \in \F, \forall k\in \T_i: 
                \frac{\co_{ijk}}{p_j + \sum_{S : k \in S} \alpha_i(S)} \le \frac{\sum_{j \in \F, k'\in \T_i} \co_{ijk'} z_{ijk'}}{w_i}$. 
    \item\label{egKKT6} $\displaystyle \forall i\in \A', \forall j \in \F: z_{ijk} > 0 \implies 
            \frac{\co_{ijk}}{p_j + \sum_{S : k \in S} \alpha_i(S)} 
            = \frac{\sum_{j \in \F, k\in \T_i} \co_{ijk'} z_{ijk'}}{w_i}$.
\end{enumerate}
In \ref{egKKT5} and \ref{egKKT6}, we have divided the conditions by  $p_j + \sum_{S : k \in S} \alpha_i(S)$ and multiplied by $\sum_{j \in \F, k'\in \T_i} \co_{ijk'} z_{ijk'}$. By the feasibility assumption, both these must be positive.

We say that $(p,\alpha)$ are \emph{optimal Lagrange multipliers} if they satisfy \ref{egKKT1}--\ref{egKKT6} together with any optimal solution $z$ to \eqref{prog:explicitReduced}.
\begin{claim}\label{claim:uncrossing}
There exists an optimal solution $z$ with optimal Lagrange multipliers $(p,\alpha)$ with the following property: 
for every agent $i\in \A'$, the support of the vector $\alpha_i$ is a chain of sets 
$S^{(i)}_1\subset \dots \subset S^{(i)}_{h_i} \subseteq \T_i$ for some $h_i \in \N$. 
\end{claim}
\begin{claimproof}
We use a standard uncrossing argument.
Let $z$ be an optimal solution to \eqref{prog:explicitReduced}. 
Let us consider the set of optimal Lagrange multipliers $(p,\alpha)$. 
For a fixed $z$, the set of vectors $(p,\alpha)$ satisfying the constraints  \ref{egKKT1}--\ref{egKKT6} forms a polytope, since each constraint can be equivalently written as a linear constraint, and \ref{egKKT3}, \ref{egKKT4}, and \ref{egKKT6}
imply boundedness.
Thus, there exists a solution $(p,\alpha)$ that maximizes the objective 
\[
\varphi(p,\alpha):=\sum_{i\in \A'}\sum_{S\subseteq \T_i}|S|^2\alpha_i(S)\, .
\]
We claim that such a solution satisfies the conditions. 
This follows by showing that for each $i\in \A'$, if $\alpha_i(X), \alpha_i(Y)>0$ then either $X\subseteq Y$ or $Y\subseteq X$. 

For a contradiction, assume $X\setminus Y, Y\setminus X\neq\emptyset$, and let $\varepsilon:=\min\{\alpha_i(X),\alpha_i(Y)\}>0$. 
Let us define $\alpha'$ as follows:
\begin{itemize}
\item  $\alpha'_i(X\cup Y)=\alpha_i(X\cup Y)+\varepsilon$;
\item  $\alpha'(X) = \alpha(X) - \varepsilon$ and $\alpha'(Y) = \alpha(Y)-\varepsilon$;
\item if $X\cap Y\neq\emptyset$, then $\alpha'_i(X\cap Y)=\alpha_i(X\cap Y)+\varepsilon$;
\item if $S\notin \{X,Y,X\cup Y, X\cap Y\}$ then $\alpha'_i(S)=\alpha_i(S)$; and
\item if $j\neq i$ then $\alpha'_j(S)=\alpha_j(S)$ for all $S$.
\end{itemize}
We claim that $(p,\alpha')$ are also optimal Lagrange multipliers. 
This gives a contradiction, since
$\varphi(p,\alpha')>\varphi(p,\alpha)$. 
Constraints \ref{egKKT1}--\ref{egKKT3} are immediate. 
Constraints \ref{egKKT5} and \ref{egKKT6} follow since $\sum_{S : k \in S} \alpha'_i(S)=\sum_{S : k \in S} \alpha_i(S)$
holds for all $i\in \A'$ and all $k\in \T_i$. 
Finally,  \ref{egKKT4} follows by observing that for any $i\in \A'$ and any $j\in \F$,
\[
\begin{aligned}
&\sum_{j \in \F, k \in X} z_{ijk} +\sum_{j \in \F, k \in Y} z_{ijk} = \f_i(X)+\f_i(Y)
\ge \f_i(X\cup Y)+\f_i(X\cap Y)\\
&\ge \sum_{j \in \F, k \in X\cap Y} z_{ijk} +\sum_{j \in \F, k \in X\cup Y} z_{ijk}
=\sum_{j \in \F, k \in X} z_{ijk} +\sum_{j \in \F, k \in Y} z_{ijk}\, , 
\end{aligned}
\]
using the submodularity of $\f_i$. 
We must have equality throughout, implying \ref{egKKT4} for $S=X\cup Y$ and $S=X\cap Y$.
\end{claimproof}

\begin{proof}[Proof of Lemma~\ref{lemma:rationalOptimum}]
Let $z$ be any optimal solution to \eqref{prog:explicitReduced} and let $(p, \alpha)$ be any optimal Lagrange multipliers as in Claim~\ref{claim:uncrossing}, with $\alpha_i$ supported on the chain $S^{(i)}_1\subset S^{(i)}_2\subset\ldots\subset S^{(i)}_{h_i}$.

Let $\F'\subseteq \F$ be the subset of goods with $p_j>0$, and let $\E'_i\subseteq \E_i$ be the set of edges $(j,k)$ for which  ${\co_{ijk}}/{(p_j + \sum_{S : k \in S} \alpha_i(S))}$ is maximized. Clearly, $z_{ijk}>0$ only if $(j,k)\in \E'_i$.

We perform the following variable substitution:
\begin{equation}\label{eq:substitute}
q_j := \frac{{1}}{p_j}\quad\forall j\in \F,\qquad \mbox{and}\qquad  
Q^{(i)}_{j\tcnt} := \frac{1}{p_j + \sum_{b = \tcnt}^{h_i} \alpha_i\left(S^{(i)}_b\right)}\quad\forall i\in \A', \ \forall\tcnt \in [h_i]\, .
\end{equation}

We show that, provided the supports $\F'$, $\E'_i$, we can define a linear program in the variables
 $q_j$'s, $Q^{(i)}_{j\tcnt}$'s, and $z_{ijk}$ as follows. We include all feasibility constraints on $z_{ijk}$ from \eqref{prog:explicitReduced} and the following additional constraints:
\begin{equation*}
  \begin{aligned}
  \sum_{i \in \A', k\in \T_i} z_{ijk} &= 1 &&\quad \forall j\in \F'\\
  \sum_{j \in \F, k \in S} z_{ijk} & = \f_i(S)&&\quad \forall i\in A',\forall S\subseteq \T_i\\
  w_i\co_{ijk}Q^{(i)}_{j\tcnt}&\le \sum_{j \in \F, k'\in \T_i} \co_{ijk'} z_{ijk'}
  &&\quad \forall i\in \A', \forall (j,k)\in\E_i,
\mbox{ and } \tcnt \mbox{ s.t. } k\in S^{(i)}_{\tcnt}\setminus S^{(i)}_{\tcnt-1}\\
 w_i\co_{ijk}Q^{(i)}_{j\tcnt}&= \sum_{j \in \F, k'\in \T_i} \co_{ijk'} z_{ijk'}
  &&\quad \forall i\in \A', \forall (j,k)\in\E_i',
  \mbox{ and } \tcnt \mbox{ s.t. } k\in S^{(i)}_{\tcnt}\setminus S^{(i)}_{\tcnt-1} \\
  Q^{(i)}_{j\tcnt}&\le Q^{(i)}_{j(\tcnt+1)}&&\quad \forall i\in \A', j\in \F', \tcnt\in [h_i-1]\\
  q_j&=0&&\quad \forall j\in \F\setminus \F'\\
 z_{ijk}&=0&&\quad \forall i\in \A', (j,k)\in\E_i\setminus \E'_i\\ 
Q,q&\ge 0
    \end{aligned}
    \end{equation*}
    Let $P\in \R^{(\sum_{i\in \A'} |\E_i|)\times \F'\times (\sum_{j\in F'} h_i)}$ be the set of feasible solutions to this LP.
According to \ref{egKKT1}--\ref{egKKT6}, $(z,q,Q)\in P$, where $(q,Q)$ is obtained from $(p,\alpha)$ as in \eqref{eq:substitute}.
Conversely, if $(z',q',Q')\in P$, then  we can map $(q',Q')$ to a nonnegative $(p',\alpha')$ such that \eqref{eq:substitute} holds and $(z',p',\alpha')$ satisfy  \ref{egKKT1}--\ref{egKKT6}.

Since all coefficients in the system are rational numbers from the input, and the feasible region $P$ is bounded, it follows that $P$ is a polytope where all basic feasible solutions are rational vectors with encoding size polynomially bounded in the input. 

Let us fix $(q',Q')$ in a basic feasible solution, and let $P''=\{z'': (z'',q',Q')\in P\}$. Then, $z''\in P''$ if and only if $z''$ is optimal with respect to \eqref{prog:explicitReduced}.
Further, $P''$ is a polytope  defined by linear constraints with polynomially bounded coefficients. Thus, the claim follows.
\end{proof}

\paragraph{The Ellipsoid Method for Rational Polyhedra}
We quickly recall some relevant concepts for the Ellipsoid Method from the book \cite{gls} by Gr\"otschel, Lov\'asz, and Schrijver. A \emph{strong separation oracle} for the convex set $K\subseteq \R^n$ takes as input a vector $x\in \R^n$, and either returns the answer $x\in K$, or returns a vector $a\in\R^n$ such that $\langle a, x\rangle>\max\{\langle a,z\rangle:\, z\in K\}$.

Let us recall the definitions of facet and
vertex complexity. We only include the definitions for polytopes, instead of general polyhedra.

\begin{definition}[{\cite[Definition (6.2.2)]{gls}}]
Let $P\subseteq \R^n$ be a polytope.
\begin{enumerate}
\item We say that $P$ has
\emph{facet-complexity at most $\varphi$}, if $P$ can be defined by a
system of linear inequalities with rational coefficients such that each inequality has encoding
length at most $\varphi$.  If $P=\R^n$, we require $\varphi\ge n+1$.
The triple $(P; n,\varphi)$ is
called a \emph{well-described polytope}.
\item We say that $P$  has \emph{vertex-complexity at most $\nu$}, if
$P$ is the convex hull of a finite set of rational vectors, all having encoding length at most $\nu$.
  $P=\emptyset$, then we require $\nu\ge n$.
\end{enumerate}
\end{definition}
\begin{lemma}[{\cite[Lemma (6.2.4)]{gls}}]\label{lem:vertex-facet} If $P$ has vertex-complexity at most $\nu$, then $P$ has facet-complexity at most $3n^2\nu$.
\end{lemma}

\begin{theorem}[{\cite[Theorems (6.4.9), (6.5.7)]{gls}}]\label{thm:well-described} For a well-described polyhedron $(P;n,\varphi)$ given by a strong separation oracle, there exists  oracle-polynomial time algorithm that either returns a vertex solution $x\in P$, or concludes that $P=\emptyset$. Given a linear objective function $\langle c,x\rangle$, if $P\neq \emptyset$ then there exists an oracle-polynomial time algorithm that finds an optimal vertex solution to $\max~\langle c,x\rangle$ s.t. $x\in P$.
\end{theorem}

An oracle-polynomial time algorithm means that the number of arithmetic operations and calls to the strong separation oracle is bounded as $\textrm{poly}(\varphi)$; note that $\varphi\ge n$.

\begin{proof}[Proof of Lemma~\ref{lemma:optimalSolution}]
Let $P$ be the set of feasible solutions and $P^*$ the set of optimal solutions to \eqref{prog:explicitReduced}. We note that $P\neq\emptyset$ since $z=0$ is a feasible solution. Further, $P^*\neq \emptyset$ since $P$ is bounded.
Lemma~\ref{lemma:rationalOptimum} asserts that this is a nonempty polytope with vertex-complexity $\mathrm{poly}(|\A|,|\G|,T,\log C, \log U)$; thus $(P^*,\sum_{i\in \A}|E_i|,\varphi)$ is a well-described polytope for some $\varphi\in \mathrm{poly}(|\A|,|\G|,T,\log C, \log U)$ by Lemma~\ref{lem:vertex-facet}.

 We now describe  the strong separation oracle to $P^*$.
For a vector  $z\in \R^{\times_{i\in \A} E_i}$, we first check whether $z\in P$. Checking the first set of $|\A|$ constraints is straightforward. The submodular constraints can be verified by solving  $|\A|$ submodular function minimization problems. We either conclude $z\in P$, or obtain a separating hyperplane for $z$ and $P$ that is also a separating hyperplane for $z$ and $P^*$.

If $z\in P$, the we compute the gradient $\nabla f(z)$, where $f(z)$ denotes the objective function. 
We then solve the linear optimization problem $\max \langle\nabla f(z),x\rangle$ s.t. $x\in P$. 
 $(P^*,\sum_{i\in \A}|E_i|,\sum_{i\in \A}|E_i|+\log T)$ is a well-described polytope since all coefficients are 0 and 1 and the left hand side values are at most $T$. Using the strong separation oracle for $P$ we just described, the second half of Theorem~\ref{thm:well-described} shows that we can find an optimal solution $x^*\in P$ in time  $\mathrm{poly}(|\A|,|\G|,T,\log C, \log U)$.

If $\max \langle\nabla f(z),x^*\rangle=\max \langle\nabla f(z),z\rangle$, i.e., if $z$ itself is an optimal solution, then we conclude that $z\in P^*$. Otherwise,  $\langle \nabla f(z),x\rangle>\langle \nabla f(z),z\rangle$ is a valid separating hyperplane.

Thus,  by the first half of Theorem~\ref{thm:well-described}, we can find an optimal solution $x\in P^*$ in time  $\mathrm{poly}(|\A|,|\G|,T,\log C, \log U)$.

This method requires the implementation of the ellipsoid method for linear optimization inside the separation oracle. We now show that this can be easily avoided by always using the hyperplane $\langle \nabla f(z),x\rangle>\langle\nabla f(z),z\rangle$, without solving the LP.
If  $z\in P\setminus P^*$, then this is always valid, but if $z\in P^*$, then this holds with equality instead of strict inequality.

Nevertheless, we can run the ellipsoid method using the gradients as separating directions (without solving the LP). This ultimately leads to concluding $P^*=\emptyset$, since the algorithm  returns a separating hyperplane for every $z\in \R^{\times_{i\in \A} E_i}$. At this point, we consider the feasible solution $z\in P$ with the largest objective value $f(z)$ visited by the algorithm, and conclude that this solution must have been optimal. This is true since if no optimal solutions would have been visited, then every separating hyperplane we used would be a valid strong separator for $P^*$, and thus, we could not have reached the false conclusion $P^*=\emptyset$.
\end{proof}

\begin{remark}
We note that a similar argument was used by Jain~\cite[Theorem 12]{Jain2007}, showing that whenever a convex set is given with a strong separation oracle and is guaranteed to contain a point of bit-complexity at most $\nu$, then a feasible solution can be found in polynomial time, using simultaneous Diophantine approximation. Our proof leverages the stronger property that the optimal solution set $P^*$ is a well-described polytope.
\end{remark}
\subsection{Sparse solutions to Eisenberg-Gale relaxation}
\label{subsection:sparseOptimalSolution}
In this section we prove Lemma~\ref{lem:sparseSolution}. Recall that the polytope $P^*$ is the set of optimal solutions to \eqref{prog:explicitReduced} as in Lemma~\ref{lemma:rationalOptimum}. In 
 Lemma~\ref{lemma:n+2m} and Corollary~\ref{corollary:sparseX},  we show that the solution of every vertex solution of $P^*$ is sparse.
In Lemma~\ref{lem:sparseSolution} we further sparsify such a solution by losing at most half of the value for each agent.
The arguments in both steps rely on bounding the number of non-zero variables in particular linear systems. 

Consider an optimal solution $z$ for~\eqref{prog:explicitReduced} that is also a basic solution to $P^*$. According to Theorem~\ref{thm:well-described}, we can require that the optimal solution found in Lemma~\ref{lemma:optimalSolution} is a basic solution.
We define $\val^*_i :=\sum_{k \in \T_i}c_{ijk} z_{ijk}$  as the optimum utility value attained by agent $i\in \A'$; by strict convexity of the objective, these values are the same for all optimal solutions.

\begin{lemma}
\label{lemma:n+2m}
Every optimal solution $z\in P^*$ satisfies
 $|\supp (z')| \le |\A'| + 2|\F^+(z')| - |R_1| - |R_2|$, where
        \begin{equation*}
        \begin{aligned}
            \F^+ (z) &= \left\{j \in \F: \sum_{i \in \A'} \sum_{k \in \T_i} z_{ijk} > 0\right\}, \\
            R_1 &= \{ j \in \F : \exists ! \, i \in \A' \text{ such that } 0 < \sum_{k \in V_i} z_{ijk} < 1 \}, \\
            R_2 &= \{ j \in \F : \exists ! \, i \in \A' \text{ such that } z_{ijk} = 1 \text{ for some } k\in \T_i \}\, . 
        \end{aligned}
        \end{equation*}
\end{lemma}
The set $\F^+$ is the set of allocated items in $\F$ by $z$;
$R_1$ is the set of items in $\F$ each of which is allocated to one agent only, but the item is not fully allocated; and
$R_2$ is the set of items in $\F$ each of which is fully allocated to agent via single edge of the graph $(\G, \T_i; E_i)$. 
Obviously, $R_1$ and $R_2$ are disjoint.
\begin{proof}[Proof of Lemma~\ref{lemma:n+2m}]
The following LP gives a description of $P^*$. We note that this is a different description from the extended system in the proof of Lemma~\ref{lemma:rationalOptimum}: here, we can make use of the optimal values $\val^*_i$ and thus do not require the dual variables. Note that the notion of vertex solutions is independent of the describing system.
\begin{equation*}
\begin{aligned}
    && \sum_{j \in \F, k\in \T_i}    \co_{ijk} z_{ijk} &\ge \hat \val^*_i       &&\quad \forall i \in \A'\\
    && \sum_{i \in \A', k \in \T_i} z_{ijk} &\le 1          &&\quad \forall j \in \F\\
    && \sum_{j \in \F, k \in \Se} z_{ijk}  &\le \f_i(\Se)       &&\quad \forall i \in \A', \forall \Se \subseteq \T_i \\
    && z &\ge 0 \,.                                     &&
\end{aligned}
\end{equation*}
In order to prove the bound on the support of a vertex (basic feasible) solution to $P^*$, 
we upper-bound the number of linearly independent \emph{tight} constraints.
Trivially, there are at most $|\A'|$ tight constraints of the first type.
By definition of sets $\F^+$ and $R_1$ there are at most $|\F^+| - |R_1|$ tight constraints of the second type.

Let us bound the maximal number of tight submodular constraints.
By Claim~\ref{claim:uncrossing}, for each agent $i \in \A'$, 
the maximal set of linearly independent tight submodular constraints forms a chain. 
Formally, for $i \in \A'$ there exist sets $S^{i}_1 \subset S^{i}_2 \subset \dots \subset S^{i}_{h_i} \subseteq \T_i$,
such that the set of constraints $\{\sum_{j \in \G, k \in S^{i}_{t}} z_{ijk}  \le \f_i(S^{i}_t) \}_{t=1}^{h_i}$
generates all the tight submodular constraints for agent $i$.
All together, there are at most $|\A'| + |\F^+| - |R_1| + \sum_{i \in \A'} h_i$ tight constraints.

Now, let us consider an element $j \in R_2$ and let $i$ be the agent such that $z_{ijk} = 1$ for some $k \in \T_i$.
Since $\f_i$ is rank function we have $z_{ijk} = 1 = \f_i(\{k\})$.
Let $S^{i}_{b}$ be the smallest set in the $i$-th chain containing $k$.
Since $\{k\}$ is also tight we can assume that $k = S^{i}_{b} \setminus S^{i}_{b-1}$. 
Therefore, the tight inequalities corresponding to  $S^{i}_b, S^{i}_{b-1}$ and $z_{ijk} \le 1 $ 
(or equivalently $\sum_{k \in \T_i} z_{ijk} \le 1$) are not linearly independent and 
we can drop the inequality corresponding to $z_{ijk} \le 1$
from the minimal set of linearly independent tight inequalities. 
In other words, we do not have to count the inequality corresponding to $j$, for $j \in R_2$ and we can replace the term $|\F_+|$ by $|\F_+| - |R_2|$.

Further, by flow conservation we have $\displaystyle |\F^+|\ge \sum_{i\in \A',j\in \F, k\in \T_i} z_{ijk} \ge \sum_{i \in \A'} \f_i(S^{i}_{h_i}) \ge \sum_{i\in \A'} h_i$.
Thus, 
\[
 |\supp(z)| \le |\A'| + 2|\F^+| - |R_1| - |R_2| \,. \qedhere
\]
\end{proof}

\begin{corollary}
\label{corollary:sparseX}
Consider an optimal vertex  solution $y$ of~\eqref{prog:EGprod} for Rado valuations.
Then,
    $|\supp (y) | \le |\A'| + 2|\F^+(y)| - |\F_1(y)|$, where
        \begin{equation*}
        \begin{aligned}
            \F^+ (y) &= \{j \in \F: \sum_{i \in \A'}y_{ij} > 0\}, \\
            \F_1 (y) &= \{j \in \F: \exists ! i \in \A' \text{ such that } y_{ij} > 0\}.
        \end{aligned}
        \end{equation*}
\end{corollary}
\begin{proof}
The optimal vertex solution $y$ can be written as $y_{ij}=\sum_{(i,k)\in \E_i}z_{ijk}$ for a vertex solution $z$ of~$P^*$. We have  
$\left|\supp\left(z\right)\right| \le |\A'| + 2|\F^+| - |R_1| - |R_2|$.
The first condition holds by definition of $\y$.
By construction we also have $\F^+ (y) =\F^+(z) =: \F^+$.
Moreover, $R_1, R_2 \subseteq \F_1$. 

By definition of $\F_1$, $R_1$ and $R_2$; 
we have $j \in \F_1\setminus (R_1 \cup R_2)$ if and only if $j$ is allocated fully to a unique agent $i$ 
and there exist different $k_1, k_2 \in \T_i$  with $z_{ijk_1}>0$ and $z_{ijk_2}>0$.
Both variables $z_{i j k_1}$ and $z_{i j k_2}$ contribute that $\y_{ij} > 0$ for the same $i,j$. 
Thus,
\begin{equation*}
|\supp(y)| \le |\A'| + 2|\F^+| - |R_1| - |R_2| - |\F_1\setminus (R_1 \cup R_2)| = |\A'| + 2|\F^+| - |\F_1|.
\end{equation*}
\end{proof}

\paragraph{Further sparsification} We showed that any basic optimal solution to~\eqref{prog:EGprod} under Rado valuations has support of size $|\A'| + 2|\F_+| - |\F_1|$.
Next, we show that any such sparse solution can be further sparsified by losing a fraction of valuation of each agent.
The main observation is that given a feasible allocation for a Rado valuation function, all ``sub-allocations'' behave in a ``locally subadditive'' way, as explained next.

Let $y'$ be a feasible allocation and $z'$ its corresponding representation in~\eqref{prog:explicitReduced}. Our argument will scale down $\y_{ij} = \pr_{ij} \y'_{ij}$ for some $\pr_{ij}\in[0,1]$. 
We have $\val_i(\y'_i) = \sum_{j\in \F, k \in \T_i} \co_{ijk} z'_{ijk}$.
Therefore, we can write $\val_i (\y'_i) = \sum_{j \in \F} u(i,j)$ where 
$u(i,j) = \sum_{k \in \T_i} \co_{ijk} z'_{ijk}$ is the value agent $i$ gets from good $j$. 
Hence, we can represent $\y_{ij} = \pr_{ij} \y'_{ij}$ as
$\y_{ij} = \pr_{ij} \sum_{k \in \T_i} z'_{ijk}$.
Assuming $\pr_{ij} \in [0,1]$ we have
$$\val_i (\y_i) \ge \sum_{j\in \F, k \in \T_i} \pr_{ij}\cdot  \co_{ijk} z'_{ijk}
 = \sum_{j \in \F} \pr_{ij} \cdot u(i,j) \, ,$$ 
where we use the fact that whenever 
$z'$ is feasible for~\eqref{prog:explicitReduced} then so is the allocation given by $\pr_{ij} z'_{ijk}$ for $j\in \F, k\in \T_i$.
In particular, this justifies the notation $\y_{ij} = \pr_{ij} \y'_{ij}$ for $\pr_{ij}\in[0,1]$ and it holds that $\val_i(\y_i) \ge \sum_{j\in \F} \pr_{ij} u(i,j)$.
Such a property is used to prove the following lemma.

\sparseSolution*

Given a $\hat y$, 
we can transform it to a vector $\y'$ with $|\supp (\y') | \le |\A'| + 2|\F^+(\y')| - |\F_1(\y')|$ by Corollary~\ref{corollary:sparseX}.
Then, the idea is to exhibit $\pr$ such that the vector $y$ defined as $y_{ij} = \pr_{ij} \y'_{ij}$ satisfies the lemma. 
Such $\pr$ needs to preserve at least half of the value for each agent and should set at least $|\F^+| - |\F_1| - |\A'| $ values of $\y'_{ij}$ to $0$.
We can find such a $\pr$ as a basic feasible solution of a system of linear (in)equalities.  

\begin{proof}
Let $\y'$ be a solution of~\eqref{prog:EGprod} with $|\supp(\y')| \le |\A'| + 2|\F^+(\y')| - |\F_1(\y')|$, given by Corollary~\ref{corollary:sparseX}.
Let $D = \{j \in \F^+(\y') : \exists i, i', i \neq i' \text{ such that } \y'_{ij} > 0 \text{ and } \y'_{i'j} >0 \}$,
i.e., 
$D$ is the set of items in $\F^+(\y')$ allocated to at least two different agents by $\y'$.
Hence, $|D| = |\F^+(\y')| - |\F_1(\y')|$.
For each $j\in D$, let $D(j)$ be a set containing two different agents $i, i'$ getting the item $j$ in $\y'$.
Such two agents are picked arbitrarily, but fixed throughout the proof for each $j$.
Let $\A'' = \cup_{j \in D} D(j)$.

We consider the following linear system with variables $\pr$. 
The value $\pr_{ij}$ represents the fraction of $\y'_{ij}$ agent $i$ keeps.
By the above, if agent obtained $u(i,j)$ value from $\y'_{ij}$ units of $j$
then agent receives $\pr_{ij} u(i,j)$ value from $\pr_{ij} \y'_{ij}$ units of good $j$ whenever $\pr_{ij} \in [0,1]$.

\begin{equation*}
\begin{aligned}
    && \sum_{j \in D}    \pr_{ij}  u(i,j) &\ge \frac{1}{2} \sum_{j \in D}     u(i,j)         && \forall i \in \A''\\
    && \pr_{ij} + \pr_{i'j} &= 1         && \forall j \in D , \, \{i, i'\} = D(j) \\
    && \pr &\ge 0 \,.                                     &&
\end{aligned}
\end{equation*}

Let us define $\y$: set $\y_{ij} = 0$ if $\pr_{ij}=0$ and $y_{ij}=\y'_{ij}$ for all other values.
Then for any feasible $\pr$ we have
\begin{itemize}
    \item The second set of constraints together with non-negativity of $\pr$ guarantees $q_{ij} \in [0,1]$ and hence 
          we can treat the values $\val_i(\y_i)\ge \pr_{ij}\val_i(\y'_i)$ as described before the statement of the lemma.
    \item By the first set of constraints and definition of $\y$, we have \\
    $$ \val_i(\y_i) \ge \sum_{j \in D} \pr_{ij} u(i,j) + \sum_{j \in \F \setminus D} u(i,j) \ge \frac{1}{2} \sum_{j \in D} u(i,j) + \frac{1}{2} \sum_{j \in \F \setminus D} u(i,j)\ge \frac{1}{2} \val_i(\y') \,.$$    
\end{itemize}

Therefore, any feasible solution of the linear system in $\pr$ gives an allocation that satisfies the first condition of the lemma. 
Let us show that the system is indeed feasible.
Namely, setting $\pr_{ij} = \frac{1}{2}$ for all $i\in \A''$ and all $j \in D$
we see that the above system is feasible. 
Since, the system is feasible we can also find a basic feasible solution $\pr$.
By counting the number of tight constraints we show that there are at least $|\F^+(\y')| - |\F_1 (\y')| - |\A''|$ zeros in $\pr$.
Thus, allocation $\y$ defined as $\y_{ij} = \pr_{ij} \y'_{ij}$ will have support smaller by at least $|\F^+(\y')| - |\F_1(\y')| - |\A''|$.

The maximum number of (tight) constraints is obviously $|\A''| + |D|$.
Therefore, $|\supp(\pr)| \le |\A''| + |D|$.
Crucially, by the second constraint we have $\F^+(\y) = \F^+(\y')$.
Hence, we only need to compare $|\supp(y')|$ and $|\supp(y)|$.
The allocation $\y'$ has exactly $2|D|$ positive variables when restricted on $D$ and $\A''$.
On the other hand, $\pr$ and therefore $\y$ take at most $|D| + |\A''|$ non-zero values on $D$ and $\A''$. 
It follows that $\y$ has at least $|D| - |\A''|$ less positive variables than $\y'$, 
i.e., $|\supp(\y)| \le |\supp(\y')| - (|D| - |\A''|)$.
By Corollary~\ref{corollary:sparseX} and since $|\A''| \le |\A'|$ we have $|\supp(\y)|  \le  2|\A'| + 2|\F^+| - |\F_1(\y')| - |D|$.
By recalling that $|D| = |\F^+| - |\F_1(\y')|$ we get $|\supp(\y)| \le 2|\A'| + |\F^+|$.
\end{proof}

\section{Phase V: Rounding the mixed solution}
\label{section:rounding}
We present the rounding for a sparse solution of~\eqref{prog:decomposedUtility}.
We recall that by \emph{sparse} we mean a feasible solution $(\y, \pi)$ of~\eqref{prog:decomposedUtility}
satisfying:
$$\supp(\y)\le 2|\A|+|\F^+| \text{ where } \F^+ = \left\{j \in \G\setminus \Hs: \sum_{i \in \A'} y_{ij} > 0\right\}\,.$$

Such a sparse solution is rounded by setting $2|\A|$ positive variables in $y$ to $0$, i.e., a reduction of $(\y, \pi)$ and allocating the items according to the support of the reduction. 
Formally, by a \emph{reduction} of $(\y,\pi)$ we mean a mixed integer solution $(\ry,\pi)$ obtained as follows 
(see Figure~\ref{figure:reduction}).
For each item $j$ a fraction of which is allocated by $y$ (i.e., $j\in  \F^+$), 
we pick an arbitrary agent $\kappa(j)$ getting the item (i.e., $\y_{\kappa(j)j}>0)$. 
We set $\ry_{\kappa(j)j}=\y_{\kappa(j)j}$, and set $\ry_{ij}=0$ if $i\neq\kappa(j)$. 
In words, the agent $\kappa(j)$ keeps getting the same amount in reduction and no other agent receives any part of item $j$.
By the bound on $\supp(\y)$, this amounts to setting $\le 2|\A|$ values $\y_{ij}$ to 0. 
Looking at the \emph{reduction} from the agents perspective:
let $\de_i$ be the number of items agent $i$ lost by reduction, i.e., the number of items $j$ for which $\y_{ij} > 0$ 
and $\ry_{ij} = 0$. 
Then, $\sum_{i\in \A'} \de_i \le 2|\A|$.

\begin{figure}[h]
    \centering
    \includegraphics[width=0.53\textwidth]{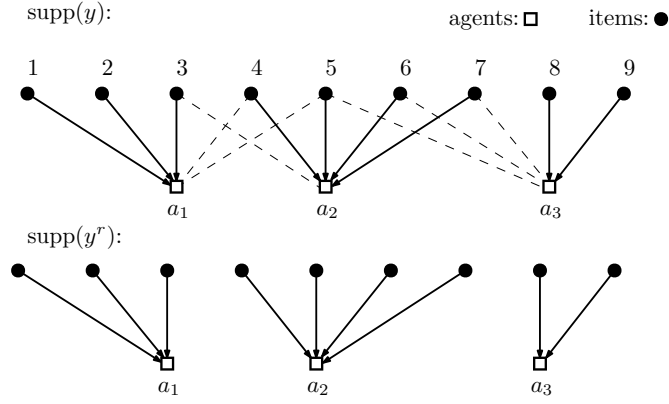}
    \caption{Support graph of an allocation $\y$. Support graph of reduction $\ry$
    obtained by $\kappa(1)=\kappa(2)=\kappa(3)= a_1$, $\kappa(4)=\kappa(5)=\kappa(6)=\kappa(7)=a_2$,  and $\kappa(8)=\kappa(9)=a_3$.
    It follows that $\de_{a_1} = 2$, $\de_{a_2}=1$ and $\de_{a_3} = 3$.}
    \label{figure:reduction}
\end{figure}

The reduction $(\ry, \pi)$ might have an arbitrarily worse objective value than $(\y, \pi)$ (e.g., if for agent $i$ we have $\val_{i\pi(i)} = 0$ and reduction sets $\ry_i = 0$),
but we show that we can find a different assignment $\rho$ such that $(\ry, \rho)$ is only worse by a constant factor than $(\y, \pi)$, 
no matter how the reduction is carried out.
The assignment $\rho$ is obtained as a combination of $\tau$ (the assignment obtained in Phase I) and $\pi$.

For a fixed reduction and the values $\de_i$, $\rho$ and its properties are given by the following lemma.

\begin{lemma}[Key rounding lemma]\label{lemma:rounding}
Let $\Hs$ be the set of most preferred items, 
$(y, \pi)$ a feasible solution to~\eqref{prog:decomposedUtility}, and let $d_i \in \N,  (d_i \ge 1)$ for each $i \in \A$.
In $O(|\A|)$ time, we can find an assignment $\rho$ such that 
$$ \overline \nsw(\y, \rho) \ge 
\frac{1}{2} \left( \prod_{i \in \A}(\de_i + 1)^{-w_i} \right)^{1/\sum_{i \in \A} w_i}
\overline \nsw(\y, \pi)$$
 and for each $i\in \A$ it holds either 
\begin{enumerate}[label = (\alph*)]
  \item\label{easyCase} $\val_{i \rho(i)} \ge \frac{1}{\de_i} \val_i(\y_i)$, or
  \item\label{harderCase} for each $j \in \F$ it holds 
        $\val_{ij} \le \frac{1}{\de_i + 1} ( \val_i(\y_i) + \val_{i \rho(i)})$.
\end{enumerate}
\end{lemma}

Intuitively, the above lemma states that starting with a feasible allocation $y$, 
we can find an assignment $\rho$ that might have smaller $\overline \nsw(y, \rho)$ than $\overline \nsw(y, \pi)$ 
but has the following nice property for each agent $i\in \A$:

\begin{itemize}
\item In case~\ref{easyCase}, $i$ values the item $\rho(i)$ at least as she values a $1/d_i$ fraction of $y_i$ 
(and thus at least a $1/(d_i + 1)$ fraction of  $\val_i(\y_i) + \val_{i \rho(i)}$).
Hence, agent $i$ keeps a $1/(d_i+1)$-fraction of her value just by keeping $\rho(i)$ even if we can take away all items $i$ gets from $\F$.
\item In case~\ref{harderCase}, every item $\F$ has a small value for $i$ when compared to the combined value of $y_i$ and $\rho(i)$.
That is, $i$ values $\y_i$ and $\rho(i)$ significantly more than any $d_i$ items combined from $\F$.
Looking at it from the other side, even if we were to take away any $d_i$ in $\F$ items from $i$ she will still keep a fraction of the value. 
\end{itemize}
The essence of both cases is that the reduction will not hurt the agent too much. 
Before we present the proof of Lemma~\ref{lemma:rounding}, we show that this is enough to prove Lemma~\ref{lem:newMatchingMain}.

\newMatchingMain*

\begin{proof}[Proof of Lemma~\ref{lem:newMatchingMain}]
We first prove the lemma for the general case.
Let $y^r$ be any reduction of $y$ and let $d_i$ be the number items agent $i$ lost in reduction. 
By sparsity in Theorem~\ref{theorem:approx-sparse} we have $\sum_{i\in \A} \de_i \le 2|\A|$.

We use Lemma~\ref{lemma:rounding} to obtain $\rho$.
Note that Lemma~\ref{lemma:rounding} requires $d_i \ge 1$ so we define $\overline \de_i = \max\{1, d_i\}$. Thus, now we have the bound 
$\sum_{i\in \A} (\overline \de_i+1) \le 4|\A|$.
Let $\rho$ be the matching obtained by Lemma~\ref{lemma:rounding} given  $\overline \de_i$'s and $y$.
By Lemma~\ref{lem:productBound} we have that 
$$
\left( \prod_{i \in \A}(\overline \de_i + 1)^{-w_i} \right)^{1/\sum_{i \in \A} w_i} \ge \frac{1}{4 \g} \,.
$$
Thus, $\overline \nsw (\y, \rho) \ge \frac{1}{8 \g} \overline \nsw (\y, \pi)$.
By the same inequality, it suffices to show that 
$\overline \nsw (\ry, \rho) \ge  \left( \prod_{i \in A}(\overline \de_i + 1)^{-w_i} \right)^{\sum_{i \in \A} w_i} \overline \nsw (\y, \rho)$.
We do so, by showing that for each $i\in \A$ it holds $\val_i(\ry_i) + \val_{i \rho(i)} \ge \frac{1}{\overline \de_i + 1} ( \val_i(\y_i) + \val_{i \rho(i)})$. By Lemma~\ref{lemma:rounding} for agent $i$ we have either~\ref{easyCase} or~\ref{harderCase}.

\begin{enumerate}
\item[~~\ref{easyCase}] In this case we have $\overline \de_i \val_{i \rho(i)} \ge \val_i(\y_i)$.
Thus, $\val_{i \rho(i)} \ge \frac{1}{\overline \de_i + 1} ( \val_i(\y_i) + \val_{i \rho(i)})$.
Consequently, $\val_i(\ry_i) + \val_{i \rho(i)} \ge \frac{1}{\overline \de_i + 1} ( \val_i(\y_i) + \val_{i \rho(i)})$.
\item[~\ref{harderCase}] We have $\val_{ij} \le \frac{1}{\overline \de_i + 1} (\val_i(\y_i) + \val_{i \rho(i)})$
for all $j\in \F$.
Denote with $\De_i$ the set of $\de_i$ items $j$ for which $\y_{ij} > 0$ and $\ry_{ij} = 0$.
By subadditivity $\val_i({\De_i}) \le \sum_{j \in \De_i} \val_{ij}$.
Therefore, $\val_i({\De_i}) \le \frac{\de_i}{\overline \de_i + 1} (\val_i(\y_i) + \val_{i \rho(i)}) \le 
\frac{\overline \de_i}{\overline \de_i + 1} (\val_i(\y_i) + \val_{i \rho(i)})$.
Hence, $ \val_i(\y_i)- \val_i({\De_i}) + \val_{i \rho(i)} \ge \frac{1}{\overline \de_i + 1} (\val_i(\y_i)+ \val_{i \rho(i)} )$.
By subadditivity and monotonicity we have $\val_i(\ry_i) \ge \val_i(\y_i) - \val_i({\De_i})$, 
proving in this case as well that $\val_i(\ry_i) + \val_{i \rho(i)} \ge \frac{1}{\overline \de_i + 1} ( \val_i(\y_i) + \val_{i \rho(i)})$.
The lemma follows.
\end{enumerate}

For additive valuations, we recall Theorem~\ref{thm:linearEG}.
It gives us an optimal solution of~\eqref{prog:EGprod} that is supported on a forest in which each tree contains an agent. 
In particular, this implies a nice property for the reductions of $\y$.
Namely, we can choose a reduction $\ry$ in which $\de_i \le 1$ for each agent $i\in \A$.
Such a reduction is obtained by rooting each tree of the forest at an arbitrary agent and letting $\kappa(j)$ to be the parent agent of item $j$.
Informally, each agent loses at most one item. Therefore, $\overline \de_i = 1$ for all $i\in \A$.
The lemma follows by Lemma~\ref{lemma:rounding}.
\end{proof} 

The proof of Lemma~\ref{lemma:rounding} is presented in the following section. 

\subsection{Constructing the new matching}
\label{section:newRematching}
Recall~\ref{phase1} where we defined $\tau$ as an assignment maximizing 
$\left( \prod_{i\in \A} \val^{w_i}_{i{\tau(i)}} \right)$
and  $\Hs$ the set 
of items assigned  by $\tau$.
We number the agents $\A=\{1,2,\ldots, n\}$, and
renumber the items $\Hs=\{1,2,\ldots,n\}$
such that $\tau = \{(i,i):\, i\in A\}$.
In other words, $\tau$ assigns item $i\in \G$ to agent $i\in \A$.

\paragraph{Intuition} 
We are given a feasible solution $(\y, \pi)$ of~\eqref{prog:decomposedUtility} and $\tau$.
For the sake of illustration assume that by using the matching $\tau$ instead of $\pi$ we don't lose too much in the objective, i.e.,
$$ \overline \nsw(\y, \tau) \ge \left( \prod_{i \in \A}(\de_i + 1)^{-w_i} \right)^{1/\sum_{i \in \A} w_i} \overline  \nsw(y, \pi) \,.$$
In this case, each agent $i$ gets the item $i$ from $\Hs$. 
Let us show that under the above assumption we can set $\rho = \tau$, i.e., 
that for each agent $i$ either~\ref{easyCase} or~\ref{harderCase} holds.
\begin{claim}\label{claim:rho=sigma}
Let $i \in \A$. Then either $\val_{i i} > \frac{1}{\de_i} \val_i(\y_i)$  or for any $j\in \F$
it holds $ \val_{ij} \le \frac{1}{\de_i + 1} (\val_{ii}+\val_i(\y_i))$
\end{claim}
\begin{claimproof}
By the optimality of $\tau$ it then holds $\val_{ii} \ge \val_{ij}$ for all $j\in \F$. 
If $\val_{ii} \ge \frac{1}{d_i} \val_i(y_i)$ then~\ref{easyCase} holds.
Otherwise, we have that $\de_i\val_{ii} < \val_i(y_i)$. 
Combining it with $\val_{ij} < \val_{ii}$, we have that 
\[(\de_i+1) \val_{ij} \le (\de_i + 1) \val_{ii} < \val_i(y_i) + \val_{ii}
= \val_i(y_i) + \val_{i\tau(i)}\,.\qedhere \]
\end{claimproof}

Therefore, our goal is to construct $\rho$ by ``replacing'' as much of $\pi$ with $\tau$ without losing too much in the objective.
By Claim~\ref{claim:rho=sigma} for any agent for which $\rho(i) = \tau(i)$ we will have either~\ref{easyCase} and~\ref{harderCase}.
We formalize this idea below, and give a way of constructing $\rho$ such that even when $\rho(i) \neq \tau(i)$ still we have either~\ref{easyCase} and~\ref{harderCase}. 

\paragraph{Algorithm} Let $(\y, \pi)$ be a feasible solution of~\eqref{prog:decomposedUtility}.
We denote with $\ls_i$ the value agent $i$ gets in $\y$, i.e., $\ls_i = \val_i(\y_i)$. 
We construct new assignment $\rho$ by combining $\pi$ and $\tau$.
In particular, whenever $\pi(i) = \tau(i)$ then we set $\rho(i) := \pi(i) = \tau(i)$
and otherwise exactly one of the following will be the case: 
$\rho(i) = \tau(i)$, $\rho(i) = \pi(i)$ or $\rho(i) = \emptyset$. 
Notation $\rho(i) = \emptyset$ represents the case that $i$ is not allocated any item from $\Hs$.
(Formally, we can allocate one item to each agent since $|\Hs|= |\A|$ but as some agents might value some items at $0$ 
it is simpler to say that agent is not allocated an item by $\rho$.)

Consider the symmetric difference of the two assignments $\pi \Delta \tau$. 
Each component is an alternating cycle; we consider the components one-by-one.
Take any component $C$ of $\pi \Delta \tau$ with $\kc$ agents and $\kc$ items. 
Let the agents in the component be $a_1,a_2,\ldots,a_\kc$. 
The numbering is modulo $\kc$: $a_{\kc+k}=a_{k}$ for all $k\in \mathbb{Z}$. 
By the convention on the numbering, 
the corresponding items are also numbered
$a_1,a_2,\ldots,a_\kc$, and $(a_k,a_k)\in \tau$ for all $k\in[\kc]$. 
We order the agents around the cycle such that $(a_k,a_{k-1})\in \pi$ for all $k\in [\kc]$. 
Let $\B := \B(C) = \{t \in [\kc] :  \ls_{a_t} > \de_{a_t} \val_{a_t a_{t-1}}\}$.
We consider two cases based on the size of $\B$:

\begin{enumerate}
\item[$|\B| = 0$.] In this case we set $\rho(a_t) = \pi(a_t) = a_{t-1}$ for all $t \in [c]$.

\item[$|\B| \ge 1$.]
First, we \emph{trim} $\pi$ by setting $\pi(a_t) = \emptyset$ for each $t \in \B$.
We have $\frac{\ls_{a_t} + \val_{a_t a_{t-1}}}{\ls_{a_t}} \le 2$ for each $t\in \B$ since $d_{a_t} \ge 1$.
In words, each agent losses at most half of her value.

After trimming $\pi$, the connected component $C$ decomposes into several alternating paths, see Figure~\ref{figure:reversing}.
Consider one such path, starting in agent $a_k$ and ending in item $a_r$.
It follows that $k \in \B$ and $t\not \in \B$ for all $k < t  \le r$.
We consider the following ratio that measures the change in the objective value by augmenting $\pi$ over the previously mentioned path:
\begin{equation*}\label{eq:ratio2}
    \varphi(C,k,r):=  \left( \frac{\ls_{a_k}}{\val_{a_k a_k}+\ls_{a_k}} \right)^{w_{a_k}}
                        \prod_{t=k+1}^{r} \left( \frac{\val_{a_t a_{t-1}}+\ls_{a_t}}{\val_{a_t a_t}+\ls_{a_t}} \right)^{w_{a_t}}\, .
\end{equation*}
    If it holds that $\varphi(C,k,r) \le \prod_{t=k}^{r-1} (\de_{a_t} + 1)^{w_{a_t}}$ 
    then we say that the \emph{interval} $[k,r]$ is \emph{reversible}.
    Moreover, we set $\rho(a_t) = \tau(a_t) = a_t$ for all $k \le t  \le r$.
    If $[k,r]$ is \emph{not} reversible then we set $\rho(a_k) = \emptyset$ and  $\rho(a_t) = \pi(a_t) = a_{t}$ for all $k< t\le r$.
    We do the same for every augmenting path. 
\end{enumerate}

\begin{figure}[h]
    \centering
    \includegraphics[width=0.9\textwidth]{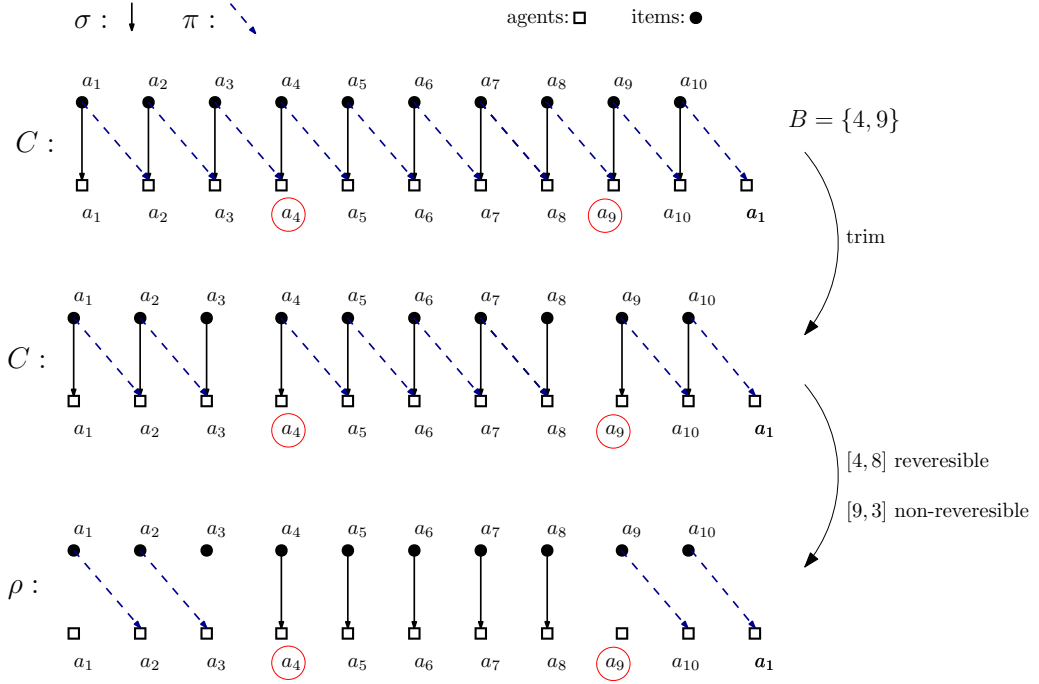}
    \caption{Assignments $\tau, \pi$, and $\rho$ resulting from $\B = \{4, 9\}$ and reversible interval $[4,8]$.}
    \label{figure:reversing}
\end{figure}

To prove Lemma~\ref{lemma:rounding}, we first show that by changing the assignment from $\pi$ to $\rho$ 
the objective value of~\eqref{prog:decomposedUtility} cannot decrease by too much.

\begin{lemma}
\label{lemma:matchingTau}
The assignment $\rho$ can be constructed in linear time (in $n$), 
and it holds $$\displaystyle \frac{ \overline \nsw(\y, \pi) }{ \overline \nsw(\y, \rho) } \le 
 2 \cdot \left( \prod_{i \in \A}(\de_i + 1)^{w_i} \right)^{1/\sum_{i \in \A} w_i}\,.$$
\end{lemma}
\begin{proof}
It suffices to prove the lemma for each of the connected components $C$ of $\pi \Delta \tau$.
For $|\B| = 0$ the lemma holds trivially. 
So assume that $|\B| \ge 1$ for the rest of the proof.

The procedure terminates in linear time, as we only require one pass through the agents and items in $C$.
To prove the bound on $\frac{ \overline \nsw(\y, \rho) }{ \overline \nsw(\y, \pi)}$,
we show that for every interval $[k,r]$ the objective value ``before averaging'' decreases at most by factor 
$2^{w_{a_k}}\prod_{t=k}^{r} (\de_{a_t} + 1)^{w_{a_t}}$.

If interval $[k,r]$ is not reversible, 
then the change in the objective function is captured by $\left(\frac{\val_{a_k a_{k-1} } + \ls_{a_k}}{\ls_{a_k}}\right)^{w_{a_k}}$,
as for every agent $a_t$ with $t\in [k+1,r]$, we have $\rho(a_t) = \pi(a_t)$, and $\rho(a_k) = \emptyset$.
Since $k\in \B$, it follows that $\ls_{a_k} > d_{a_{k}} \val_{a_k a_{k-1}} \ge \val_{a_k a_{k-1}}$.
Thus, $\left(\frac{\val_{a_k a_{k-1}} + \ls_{a_k}}{\ls_{a_k}}\right)^{w_{a_k}} < 2^{w_{a_k}}$.

\medskip
If, on the other hand, $[k, r]$ is reversible, then the difference in the objectives is captured by 
\[
    \left(\frac{\val_{a_k a_{k-1}} +  \ls_{a_k} }{\val_{a_k a_k}+\ls_{a_k}} \right)^{w_{a_k}} 
    \prod_{t=k+1}^{r} \left( \frac{\val_{a_t a_{t-1}}+\ls_{a_t}}{\val_{a_t a_t}+\ls_{a_t}} \right)^{w_{a_t}} =
    \left(\frac{\val_{a_k a_{k-1}} +  \ls_{a_k} }{\ls_{a_k}} \cdot \frac{\ls_{a_k} }{\val_{a_k a_k}+\ls_{a_k}} \right)^{w_{a_k}} 
    \prod_{t=k+1}^{r} \left( \frac{\val_{a_t a_{t-1}}+\ls_{a_t}}{\val_{a_t a_t}+\ls_{a_t}} \right)^{w_{a_t}} 
\]
As $[k,r]$ is reversible $\varphi(C, k, r) = \left(\frac{\ls_{a_k}}{\val_{a_k a_k}+\ls_{a_k}} \right)^{w_{a_k}}
    \cdot \prod_{t=k+1}^{r} \left( \frac{\val_{a_t a_{t-1}}+\ls_{a_t}}{\val_{a_t a_t}+\ls_{a_t}} \right)^{w_{a_t}} 
    < \prod_{t=k}^{r} (\de_{a_t} + 1)^{w_{a_t}}$.
    Since $k \in \B$ and $d_{a_k} \ge 1$ we again have 
    $\frac{\val_{a_k a_{k-1}} + Y_{a_k}}{\ls_{a_k} } < 2$.
    Hence, the change in the objective value is  bounded by 
    $ 2^{w_{a_k}} \cdot \prod_{t=k}^{r} (\de_{a_t} + 1)^{w_{a_t}}$.
\end{proof}

It is left to show that for each agent $i$ either~\ref{easyCase} or~\ref{harderCase} holds.
Recall that $\ls_i = \val_i(y_i)$. 
\begin{lemma}\label{lemma:preserve}
Let $i \in \A$. Then we either have 
\begin{enumerate}
    \item[~\ref{easyCase}] $\val_{i \rho(i)} \ge \frac{1}{\de_i} \val_i(\y_i)$, or
    \item[~\ref{harderCase}] for each $j \in \F$ it holds 
        $\val_{ij} \le \frac{1}{\de_i + 1} ( \val_i(\y_i) + \val_{i \rho(i)})$.
\end{enumerate}
\end{lemma}

To prove the lemma we use the following simple claim, which can applied to any agent $i\not \in B$:
\begin{claim}\label{cl:upper-u}
For any agent $i\in \A$, if $\ls_{i}\le \de_i \val_{i \pi(i) }$, then 
$
\displaystyle \frac{\val_{i \pi(i)}+\ls_i}{\val_{ii}+\ls_i} \le \frac{(\de_i + 1)\val_{i  \pi(i)}}{\val_{ii}}\, .
$
\end{claim}
\begin{proof}[Proof of Lemma~\ref{lemma:preserve}]
If $\rho(i)=i$, that is, agent $i$
receives the same item in $\rho$ as in $\tau$ then the lemma follows by Claim~\ref{claim:rho=sigma}.
For the rest of the proof we assume $\rho(i) \neq i$.
Hence, either $\rho(i) = \pi(i)$ or $\rho(i) = \emptyset$.

We consider the component $C$ of $\tau \Delta
\pi$ containing an agent $i$.
We use the notation as before, denoting the
agents in $C$ by $a_1,a_2,\ldots,a_\kc$, and letting $i=a_k$. 

\medskip
If $\rho(a_k) = \pi(a_k) = a_{k-1}$ then for $i$ it holds~\ref{easyCase}. 
Namely, $\rho(a_k) = a_{k-1}$ implies that $k \not \in \B$ as otherwise this would be trimmed.
Thus $\ls_{a_k} \le d_{a_k} \val_{a_k a_{k-1}}$; 
or equivalently $\val_{a_k a_{k-1}} \ge \frac{1}{d_{a_k}} \ls_{a_k}$. 

\medskip
If on the other hand $\rho(a_k) = \emptyset$, we have that $k \in \B$ 
and also that the interval $[k, r]$ with starting and $k$ and ending in $r$
that corresponds to some alternating path in $C$ is \emph{not} reversible (otherwise, $\rho(a_k) = a_k$).
Therefore, $\varphi(C,k,r) > \prod_{t=1}^{r} (\de_{a_t} + 1)^{w_{a_t}}$. 
Recall that for each such considered interval we have $k\in \B$ and $t\not \in B$.
Starting with $ \prod_{t=k}^{r} (\de_{a_t} + 1)^{w_{a_t}} < \varphi(C,k,r)$ and then by Claim~\ref{cl:upper-u} we obtain
    \[
    \begin{aligned}
    1 &< \prod_{t=k}^{r-1} (\de_{a_t} + 1)^{-w_{a_t}}
    \cdot \left( \frac{\ls_{a_k}}{\val_{a_k a_{k}}+\ls_{a_k}} \right)^{w_{a_k}}
    \cdot \prod_{t=2}^{r} 
    \left( \frac{\val_{a_t a_{t-1}}+\ls_{a_t}}{\val_{a_t a_{t}}+\ls_{a_t}} \right)^{w_{a_t}} \\
    &\le
    (\de_{a_k} + 1)^{-w_{a_k}}\cdot 
    \left( \frac{\ls_{a_k}}{\val_{a_k a_{k}}+\ls_{a_k}} \right)^{w_{a_k}}
    \cdot \prod_{t=2}^{r} 
    \left( \frac{\val_{a_t a_{t-1}}}{\val_{a_t a_{t}}} \right)^{w_{a_t}}\, .
    \end{aligned}
    \]
We further bound
    \[
    1<(\de_{a_k} + 1)^{-w_{a_k}}\cdot 
    \left(  \frac{\ls_{a_k}}{\val_{a_k j}} 
    \cdot  \frac{\val_{a_k j}}{\val_{a_k a_k}} \right)^{w_{a_k}} 
    \cdot \prod_{t=2}^{r} 
    \left(\frac{\val_{a_t a_{t-1}}}{\val_{a_t a_{t}}} \right)^{w_{a_t}}\, .
    \]
By the optimal choice of $\tau$, for every $j \in \F$ we have
    $$1 \le \left( \frac{\val_{a_k a_k}}{\val_{a_{k} j}} \right)^{w_{a_k}}
    \cdot
    \prod_{t=2}^{r}\left( \frac{\val_{a_t a_{t}}}{\val_{a_t a_{t-1}}} \right)^{w_{a_t}} \,. $$ 
Combining the last two inequalities, we obtain $\ls_{a_k}>(\de_{a_k} + 1) \val_{a_k j}$.
Hence, in this case~\ref{harderCase} holds, by  recalling that $i=a_k$ and $\rho(a_k)=\emptyset$.
\end{proof}

\section[{Rado valuations and M-natural-concave functions}]{Rado valuations and M\nat-concave functions}
\label{section:utilityDetails}

In this section, we prove Lemma~\ref{lemma:RadoIsMconcave} showing that Rado valuations are M\nat-concave and Lemma~\ref{lemma:concaveClosure} showing that 
the function $\nu$ as defined in Theorem~\ref{thm:integral} is indeed the continuous closure of the Rado valuation. We also discuss related conjectures on constructive characterizations of M\nat-concave functions.

\subsection[Rado valuations are M-natural-concave]
{Rado valuations are M\nat-concave}
Rado valuations turn out to be a special case of a more
general construction described in \cite{Murotabook}, called \emph{`transformation by networks'}. We now present it in the special case when the network is a bipartite graph {\em (instead of a directed graph)} with linear edge costs {\em (instead of concave functions)}, and the functions are restricted to the binary domain {\em (instead of the nonnegative integer lattice)}.

\begin{theorem}[{\cite[Chapter 9.6.1]{Murotabook},\cite[Section 6.2]{murota2016discrete}}]
Let 
$H = (\G, \T'; \E')$ be a bipartite graph with cost function  $\co': \E' \to \R_+$.
Given an M\nat-concave function $g : \{0,1\}^{\T'} \to \R \cup \{- \infty\}$ the following  function $\tilde g:\{0,1\}^\G \to \R$ is also M\nat-concave:
\begin{equation}\label{prog:transformationByBipartite}
\tilde g(x) = \max_{y\in \{0,1\}^{\T'}, z\in \{0,1\}^{{\G}\times \T'}} 
\left\{ g(y) + \co'(z): 
\sum_{k \in \T'} z_{kj} = x_j \,, \forall j\in \G 
\text{ and } \sum_{j\in \G} z_{kj} = y_k \,, \forall k \in \T' \right\} \,.
\end{equation}
\end{theorem}
We use the transformation to show that the Rado valuation functions are M\nat-concave.

\RadoIsMconcave*
\begin{proof}
Consider a Rado valuation $\val$ as in Definition~\ref{def:Rado},
given by a bipartite graph $(\G, \T; \E)$ with a cost function $\co : \E \to \R$ on the edges, and a matroid $ \M = (\T, \I)$ with rank function $\f$.

Let us define $\T'=\T\cup D$ with a set $D$ of $|\G|$ dummy nodes. Let $\E'$ be the union of $\E$ and a perfect matching of dummy edges between $\G$ and $D$. Let $\co'_e=\co_e$ for all $e\in \E$ and $\co'_e=0$ on all new dummy edges.
Let us define a matroid $\M'=(\T',\I')$ where a set $S\in \I'$ if and only if $S\cap \T\in \I$; that is, we add the dummy elements in $D$ freely to the matroid.
We define 
\begin{equation}\label{def:y}
g(y) = 
\left\{
    \begin{array}{ll}
        0  & \mbox{if } y=\chi_S \mbox{ for some }S\in \I'\, , \\
        -\infty & \mbox{otherwise.} 
    \end{array}
\right.
\end{equation} 
As the indicator function of the independent sets of a matroid, $g$ is well-known to be M\nat-concave (see e.g. \cite[Section 4.7]{Murotabook}).
We claim that $\tilde g$ defined in \eqref{prog:transformationByBipartite}  equals the Rado valuation $\val$.

First, for $S\subseteq \G$, let $M$ be the maximum cost matching in the definition of $\val(S)$. We can extend this to a perfect matching $M'$ with $\delta_{S}(M')=S$ by adding dummy edges incident to the nodes in $S\setminus \delta_S(M)$.
We then define $z=\chi_{M'}$ and $y=\chi_{\delta_{\T'}(M')}$. Thus, $\tilde g(\chi_S)\ge g(y)+c'(z)=0+c(M)=\val(S)$. Conversely, consider the optimal $(y,z)$ in the definition of 
$\tilde g(\chi_S)$. By the above bound, we know that $\tilde g(\chi_S)\ge \val(S)\ge 0$ is finite, and therefore $g(y)=0$. The set of (non-dummy) edges $e\in \E$ with $z_e=1$ thus form a matching with $\delta_{\T}(M)\in \I$, $\delta_{\G}(M)\subseteq S$, and $c(M)=\tilde g(\chi_S)$, showing that $\tilde g(\chi_S)\le \val(S)$.
\end{proof}

\subsection{Concave closure of Rado valuations}
We now complete the proof of Theorem~\ref{thm:integral}, showing that the function $\nu$ defined in \eqref{prog:utility} is indeed the continuous extension of the Rado valuation $\val$.

The value $\overline g (x)$ for $x\in [0,1]^m$ is defined by a linear program \eqref{eq:bar-val}.\footnote{For M\nat-concave functions defined over the lattice $\Z^n$, the definition of the extension includes a constraint for every lattice point, thus, the system is not finite. Still, it can be described by a `local' linear program, see \cite[(3.64) and Theorem 6.42]{Murotabook}.}
 In the proof, we will use the dual LP:
\begin{equation}\label{eq:extend-dual}
\begin{aligned}[t]
    \bar \val(x)=&\min \quad p^\top x+\alpha\\
    &\begin{aligned}
        \text{s.t.:} \, && p(S)+\alpha &\ge \val(S) && \forall S\subseteq \G\\
        && (p,\alpha)   &\in \R^{m+1} &&
    \end{aligned}
\end{aligned}
\qquad\qquad
\begin{aligned}[t]
    &\max \quad  \sum_{S\subseteq G} \lambda_S \val(S)\\
    &\begin{aligned}
        \text{s.t.:} \, &&\sum_{S\subseteq G} \lambda_S \chi_S &=x &&\\
        &&\sum_{S\subseteq G} \lambda_S &=1 &&\\
        &&\lambda & \ge 0 &&
    \end{aligned}
\end{aligned}
\end{equation}

\begin{lemma}\label{lemma:concaveClosure}
Let $\val$ be a Rado valuation given by a bipartite graph $(\G, \T; \E)$ with costs on the edges $\co : \E \to \R$,
and a matroid $ \M = (\T, \I)$ with a rank function $\f=\f_\M$ as in Definition~\ref{def:Rado}. Let $\nu(x)$ be the function defined in \eqref{prog:utility}, that is,
\begin{equation*}
\begin{aligned}
   \nu(x)&:= \quad \max \quad \sum_{(j,k)\in \E}    \co_{jk} z_{jk}   \\
   &\begin{aligned}
    \text{s.t.: }   && \sum_{k\in \T } z_{jk} &\le x_j     &&\quad \forall j \in \G \\
                    &&\sum_{j \in \G, k \in T} z_{jk} &\le \f(T)        &&\quad \forall T \subseteq \T \\
                  &&z &\ge 0  \,.      &&
    \end{aligned} 
\end{aligned}
\end{equation*}

 Then, $\overline \val (x)= \nu (x)$ holds  for every $x\in \R^m_+$.
\end{lemma}
\begin{proof}
We let ${\cal M}(x)$ denote the set of feasible solutions of~\eqref{prog:utility}.
Fix any $x\in \R^m$. 
We first show that $\overline \val(x)\le \nu(x)$.

Consider an optimal solution $\lambda$ for the dual LP in
\eqref{eq:extend-dual} such that
$\overline \val(x)=\sum_{S\subseteq \G} \lambda_S \val(S)$. 
For every $S\subseteq \G$, 
we have an integral allocation $M_S$ of the goods in  ${\cal M}(\chi_S)$
that is optimal in the linear program~\eqref{prog:utility} 
defining $\nu(\chi_S)=\val(S)$; these two are equal using Theorem~\ref{thm:integral}.
It is easy to see that $\sum_{S\subseteq \G}\lambda_S M_S \in {\cal M}(x)$.
Thus, $\overline \val(x)\le \nu(x)$.

For the other direction $\overline \val(x)\ge \nu(x)$,  let $z$ be the optimal solution
defining $\nu(x)$ in~\eqref{prog:utility}.
By the integrality of the bipartite matching polytope, we can write the fractional matching $z$ as a convex combinations 
of integral allocations 
$M_S$ for $S\subseteq \G$, i.e.,
$z=\sum_{S\subseteq \G}\lambda_S M_S$ for some $\lambda\ge 0$ with $\sum \lambda_S=1$.
The dual of~\eqref{prog:utility} is
\begin{equation*}
\begin{aligned}[t]
    &\min~ \sum_{j\in \G} x_j\pi_j + \sum_{T \subseteq \T} \rho_{T}\\
    &\begin{aligned}
    \text{s.t.:} \quad && \pi_j + \sum_{T: k\in T} \rho_{T} &\ge \co_{jk}        &&\forall j\in \G, \forall T \subseteq \T\\
    && \pi \in \R^{\G}_+,\quad \rho&\in\R^{2^{\T}}_+ \,.&&
    \end{aligned}
\end{aligned}
\end{equation*}
Consider an optimal dual solution $(\pi,\rho)$. By complementarity,
$\pi_i+\sum_{\Se: k\in T} \rho_{T} = \co_{jk}$ for every $(j,k) \in \supp(z)$; if
$\rho_{T}>0$ for $T \subseteq \T$ then $z(\delta(T))=\f(T)$, and if $\pi_j>0$ for
$j\in \G$ then $z(\delta(j))=x_j$.

Since $z=\sum_S \lambda_S M_S$, we have $M_S\subseteq \supp(z)$, and 
$\delta_{M_S}(\Se)=\f(\Se)$ whenever $z(\delta(\Se))=\f(S)$. 
Further, $z(\delta(j))=x_j$ implies that every matching $M_S$ with $j\in S$ covers $j$.
We see that $\chi_{M_S}$ and $(\pi,\rho)$ satisfy complementary slackness in
\eqref{prog:utility} for every set $S$ with $\lambda_S>0$. Thus, $c(M_S) = \nu(\chi_S)$, and $\nu(\chi_S)=\val(S)$ again by Theorem~\ref{thm:integral}.
 We can thus conclude that
\[
\nu(x)=\sum_{S\subseteq \G}\lambda_S c(M_S)=\sum_{S\subseteq \G}\lambda_S \val(S)\le \bar\val(x)\, ,
\]
completing the proof. 
\end{proof}

\subsection[{Conjectures on characterizing M-natural-concave functions}]{Conjectures on characterizing M\nat-concave functions}\label{sec:conjectures}

First, we answer  Frank's question negatively, showing that  Rado valuations do not cover the entire class of M\nat-concave valuations.
Lehmann et al.~\cite[Example 1]{DBLP:journals/geb/LehmannLN06} gave an example that is an M\nat-concave (gross substitutes) valuation but not OXS.  We show that the same example is also not a  Rado valuation; the proof is similar.

\begin{lemma}\label{lem:non-Rado}
Consider the following valuation on the ground set $\G = \{1,2,3,4\}$. We define $\val(S)=10$ if $|S|=1$, and $\val(S)=19$ for all sets with $|S|\ge 2$ except $\val(\{1,3\}) =\val(\{2,4\}) = 15$. This is M\nat-concave, but not a Rado valuation function.
\end{lemma}
\begin{proof}
The proof that $\val$ is a gross substitutes/M\nat-concave valuation is given in \cite[Claim 2]{DBLP:journals/geb/LehmannLN06}. Let us show that it is not a Rado valuation.
For a contradiction, assume $\val$ is a Rado valuation as in Definition~\ref{def:Rado}. We can assume that the matroid on $\T$ does not contain any loops (rank-0 elements), and any parallel elements, i.e., any set $S\subseteq\T$ with $|S|\ge 2$    and $\f(S)=1$; we can contract any such set to a single element and obtain another representation.

Trivially, we can assume that no edge in the bipartite graph $(\G, \T; \E)$ has cost more than $10$. 
By $\val(\{1\}) = 10$ we have an element $u\in \T$ with $c_{1u} = 10$.
Since $\val(\{2\})=10$, there is $u'\in \T$ such that $c_{2u'} = 10$.
Since $\val(\{1, 2\}) < 20$ we have $u'=u$ as otherwise $(1,u), (2,u')$ would be an independent matching of cost 20, 
since $\f(\{u,u'\})=2$ by the above assumption. 

An analogous argument shows that $c_{ju} = 10$ for all $j \in \{1,2,3,4\}$.
We must have $c_{jk} \le 5$ for any $j \in \{1,2,3,4\}$ and any $k \in \T \setminus \{v\}$, as otherwise we would have an independent matching of cost $>15$ covering $\{1,3\}$ or $\{2,4\}$, again using the assumption of no parallel elements in $\T$.
Now, it is clear that we cannot realize $\val(\{1,2\}) = 19$.
\end{proof}

The reason why Rado valuations are not a rich enough class is that it is not closed under {\em endowment operations}. Given a valuation $\val:2^\G\to \R$ and a subset $T\subseteq \G$, we can define the valuation $\val':2^{\G'\setminus T}\to\R_+$ as
\[
\val'(X)=\val(X\cup T)-\val(T)\, .
\]
Using Definition~\ref{def:M-conc}, it is immediate that if $\val$ is M\nat-concave than so is $\val'$. It is not difficult to check that the example in Lemma~\ref{lem:non-Rado} arises as the endowment of a Rado valuation, showing that Rado valuations are {\em not} closed under endowment operations.

Endowment can be seen as a minor operation. Let us say that $\val$ is a Rado minor valuation if it arises from a Rado valuation by the endowment operation. Note that this class is trivially closed for endowment.
This motivates the following conjecture:
\begin{conjecture}\label{conj:Rado-minor}
Every M\nat-concave valuation arises as a Rado minor valuation.
\end{conjecture}
Ostrovsky and Paes Leme \cite{ostrovsky2015gross} previously posed the following stronger {\em ``matroid based valuation conjecture''}. We define the merging/convolution of the valuations $\val_1,\val_2:2^\G\to \R$ as
\[
\val^*(S)=\max_{T\subseteq S} \val_1(T)+\val_2(S\setminus T)\quad\forall S\subseteq\G\, .
\]
Merging two M\nat-concave functions results in an M\nat-concave function. 
\begin{conjecture}[Ostrovsky and Paes Leme \cite{ostrovsky2015gross}]\label{conj:OPL}
Every M\nat-concave valuation arises by the repeated application of endowment and merging operations starting from weighted matroid rank functions.
\end{conjecture}

This conjecture is still open. Tran~\cite{Tran2019}  showed that only allowing merging above is not sufficient, even if starting from a slightly broader class also including partition valuations.

Conjecture~\ref{conj:Rado-minor} is a natural weakening of Conjecture~\ref{conj:OPL}: weighted matroid rank functions form a subclass of Rado valuations, and it is easy to verify that Rado minor valuations are closed under merging and endowment. 

Balkanski and Paes Leme \cite{Balkanski2020} gave a negative answer to the question  whether every M\nat-concave valuation arises as a conic combination of (unweighted) matroid rank functions. Note that M\nat-concave functions are not closed under conic combinations, even the sum of two matroid rank functions may not be M\nat-concave. Thus, the questions was whether conic combinations of matroid ranks forms a superclass of the M\nat-concave valuations.

\section{Improved product bound}
\label{section:productBounds}
We show that the bound in Lemma~\ref{lem:productBound} can be slightly improved.
Throughout the section the base of the logarithm is $e$.
We recall that the Lamberth function $\cal W$ is the inverse of $t \mapsto t \ln t$ for $t \in \R_+$.
For $x > e$ it holds ${\cal W}(x) < \log x$; and for $x > 41.19$ it holds ${\cal W}(x) > \log x - \log (\log x)$, see~\cite{hassani2005approximation}.
Let 
$\psi(x) = \left(\frac{x-2}{{\cal W}(\frac{x-2}{e})} \right)^{x/\left(x-2+\frac{x-2}{{\cal W}(\frac{x-2}{e})}\right)}$ (for $x > 2$).
Then, by the above bound we get $\psi(x) \le \max \left\{\overline c, \frac{x-2}{\log(\frac{x-2}{e}) -\log\log(\frac{x-2}{e})} \right\}$ 
for some constant $\overline c$ that depends on $41.19$.
Now, we can prove our lemma. 

\begin{lemma}\label{lemma:productBound} 
Let $n \in \N$, $S \subseteq [n]$,  and $1\le w_1, \dots, w_n <  \g-1$.
For $i\in S$ let $\s_i \in \R_+$ such that $\sum_{i \in S} \s_i \le c \cdot n$.
Assuming $c$ is a constant we have $$\left( \prod_{i \in S} \s_i^{w_i} \right)^{1/\sum_{i=1}^n w_i} 
\le c \cdot \psi(\g) = O\left(\frac{\g}{\log(\g)} \right)\,.$$
\end{lemma}

\begin{proof}
\begin{sloppypar}
We present the proof for $c=1$, the general cases easily reduces to $c=1$ by scaling.   
Since ${\cal W}(x)e^{{\cal W}(x)} = x$ we have $e^{{\cal W}(\frac{x-2}{e}) + 1} = e \cdot \frac{x-2}{e} \cdot \frac{1}{{\cal W}(\frac{x-2}{e})} = \frac{x-2}{{\cal W}(\frac{x-2}{e})}$.
Hence, 
$$\left( e^{{\cal W}(\frac{x-2}{e}) + 1} \right)^{
                x/\left(x-2+e^{{\cal W}(\frac{x-2}{e}) + 1)}\right)} 
= \left(\frac{x-1}{{\cal W}(\frac{x-2}{e})} \right)^{
                x/\left(x-2+\frac{x-2}{{\cal W}(\frac{x-2}{e})}\right)}\,,$$
for $x > 2$. 
It suffices to prove that 
$$ \left( \prod_{i \in S} \s_i^{w_i} \right)^{1/\sum_{i=1}^n w_i} \le 
        \left( e^{{\cal W}(\frac{\g-2}{e}) + 1} \right)^{
                \g/\left(\g-2+e^{{\cal W}(\frac{\g-2}{e}) + 1)}\right)} \,.$$
Without loss of generality we can assume that $\s_i \ge 1$. 
Then the worst case is if $w_i = \g-1$ for all $i \in S$ and $w_i = 1$ for $i \in [n]\setminus S$.
For fixed size of $S$ ($k = |S|$), 
the product $\prod_{i\in S} \s_i^{\g-1}$ is maximized when all $\s_i$ are the same. 
Hence, $\left( \prod_{i \in S} \s_i^{w_i} \right)^{1/\sum_{i=1}^n w_i}$ is  upper-bounded by 
$\left( \frac{n}{k} \right)^{k(\g-1)/(k(\g-1) + n - k)}$.
Let $\xi = \frac{n}{k}$ then $\left( \frac{n}{k} \right)^{k(\g-1)/(k(\g-1) + n - k)} = \xi^{(\g-1)/(\g-2 + \xi)}$.
By the first order conditions, the value $\xi^{(\g-1)/(\g-2 + \xi)}$ achieves the maximum for $\xi = e^{{\cal W}(\frac{\gamma-2}{e}) + 1}$.
Hence, 
\begin{equation*}\label{eq:lossy}
\left( \frac{n}{k} \right)^{k\g/(k\g + n - k)} \le \left( e^{{\cal W}(\frac{\g-2}{e}) + 1} \right)^{(\g-1)/\left(\g-2+e^{{\cal W}(\frac{\g-2}{e}) + 1)}\right)}\,. \qedhere
\end{equation*}
\end{sloppypar}
\end{proof}

\section{Connection to spending restricted equilibrium}
\label{section:connectionToSR}

The first constant factor approximation algorithm for the Nash social welfare problem was given by~\citet{cole2015approximating} 
using the so-called \emph{spending restricted (SR)} equilibrium. 
Since then, the SR-equilibrium is one of the main concepts used in the design of the approximation algorithms for the NSW problem~\cite{anari2018nash,cole2017convex,garg2018approximating,garg2019auction}. 

An important feature of the SR-equilibrium is that the items highly valued by the agents are recognized as items with price more than $1$ (\emph{expensive}) in the equilibrium. 
Isolating such items is at the essence of the approximation algorithms in the literature.   
The main idea is that each of the expensive items must be allocated integrally to one agent only,
thereby preventing the unbounded integrality gap arising when several agents share a very desirable good, see~\citep[Lemma 3.1]{cole2015approximating}.

In this section, we illustrate a connection between the approach we use and the SR-equilibrium.
In that light, for the rest of the section we focus on the case of symmetric Nash social welfare problem where agents have additive valuations.
We show that the \emph{set of the most preferred items} $\Hs$ obtained in~\ref{phase1} contains all the expensive items in an SR-equilibrium. 
Similarly to the algorithms relying on the SR-equilibrium where expensive items have special status during rounding, 
the items in $\Hs$ are allocated integrally throughout our algorithm.
Intuitively, this is how we are overcoming the unbounded integrality gap.

\paragraph{SR-equilibrium}
We quickly recall the necessary definitions and refer the reader to~\cite{cole2015approximating} for more details.
The market consists of a set of divisible items $\G$, 
agents $\A$ each of which has a budget of $1$ and an additive valuation function over the items.
A valuation of agent $i$ is additive if her value is given as $\val_i(x_i) = \sum_{j \in \G} \val_{ij} x_{ij}$ for all $x_{i} \in \R^{\G}_{+}$ and $\val_{ij} \in \R_+$.

Consider prices $p\in \R^{\G}_+$ for the items in $\G$. 
We say that an item $j$ is \emph{maximum bang per buck} (MBB) 
for agent $i$ if $j \in \argmax_{j \in \G}\{\val_{ij}/ p_j\}$.
For an allocation $x$, the spending of an agent on $x_i$ is $p^{\top}x_i$ 
and the spending on an item $j$ is $\sum_{i\in \A} p_j x_{ij}$.
The MBB items are exactly the items an agent would buy at prices $p$ in order to maximize its valuation such that spending is not more than a given budget. 

\begin{definition}
    A spending restricted (SR) equilibrium is a fractional allocation $x$
    and a price vector $p$ such that every agent spends all of her budget on her MBB items at prices $p$,
    and the total spending on each item is equal to $\min\{1, p_j\}$.
\end{definition}
By scaling the valuation of each agent we can assume that the maximum bang per buck is one for all agents. 
Under such a scaling, in an SR-equilibrium we also have that $v_{ij} = p_j$ whenever item $j$ is MBB for agent $i$ and $v_{ij} < p_j$ otherwise.
We work with this assumption for the rest of this section.

\paragraph{NSW and SR-equilibrium} 
Consider a NSW welfare instance with items $\G$ and agents $\A$ where each agent $i$ has additive valuation. 
For the NSW problem, the valuations are discrete function and the value of a subset of items $S$ for agent $i$ is given by $\val_i(S) = \sum_{j \in S} \val_{ij}$. 
The extension of an additive valuation $\val_i$ to $\R^{\G}_+$ is naturally defined as $
\val_i(x_i) = \sum_{j \in \G} \val_{ij} x_{ij}$ for all $x_{i} \in \R^{\G}_{+}$.
We construct the market from the NSW instance from the same set of items $\G$ that are now declared divisible and the set of agents
$\A$ each equipped with the extension of the discrete additive valuation and budget one.

Let $(x, p)$ be an SR-equilibrium in such a market.
Define the set of \emph{expensive goods} $\overline H$ as $\overline H := \{j \in \G: p_j >1\}$.
\citet{cole2015approximating} proved that $\left(\prod_{j \in \overline H} p_j \right)^{1/|\A|}$ is an upper-bound on the optimal value of NSW, 
and gave a rounding algorithm that uses an SR-equilibrium as a starting point.

In the next lemma we show that $\overline H \subseteq \Hs$, where $\Hs$ is the set of the most preferred goods obtained in~\ref{phase1}
of our algorithm. 
Recall that $\tau$ is an assignment maximizing $\left(\prod_{i \in \A} \val_{i\tau(i)}\right)^{1/|\A|}$ and that $\Hs := \tau (A)$.
In words, $\tau$ maximizes the NSW welfare under the constraint that each agent gets exactly one item. 

For the purposes of the proof recall that the
the spending graph $(\A, \G; \E_x )$  of an allocation $x$ is defined as $ij \in \E_{x}$ if and only if $x_{ij} > 0$.

\begin{lemma}
It holds $\overline H \subseteq \Hs$.
\end{lemma}
\begin{proof}
Using a cycle canceling argument, we can assume that the spending graph of SR-equilibrium $(x,p)$ is a forest $F$.
Moreover, since $x$ is an SR-equilibrium allocation, every tree contains at least one agent and one item.
The next claim states that an expensive item is a leaf in some tree in $F$ only in a very special case.

\begin{claim}\label{claim1}
Let $T=(\A_1, \G_1; E_1)$ be a tree component of  $F$ and $j \in \G_1$ an item in $T$.
If $p_j > 1$ then either $|A_1| = |\G_1| = 1$ or $j$ is not a leaf of $T$. 
\end{claim}
\begin{claimproof}
    By definition of SR-equilibrium each agent spends all of her budget which is $1$.
    If $j$ is a leaf, then there is unique agent $i$ buying $j$. 
    Moreover, $i$ spends all $1$ unit of her budget on $j$ and cannot buy any other item. 
    Thus, $\A_1 = \{i\}$ and $\G_1 = \{j\}$.
\end{claimproof}

Let $\kappa : \overline H \to \A$ such that $x_{\kappa(j) j} > 0$.
Such an function $\kappa$ exists by definition of SR-equilibrium.
Moreover, by Claim~\ref{claim1} we can choose $\kappa$ to be an assignment 
(root every tree of $F$ in an arbitrary item and assign the expensive item to any child agent).
We are ready to prove the lemma.

For the sake of contradiction suppose that there is an item $j_1 \in \overline H$ such that $j_1 \not \in \Hs$.
In other words, $p_{j_1} > 1$ and $j_1$ is not allocated to any agent by $\tau$.
By definition we have $\overline H \le |\A| = |\Hs|$.
Consider the component of the symmetric difference $\tau \Delta \kappa$ containing $j$.
Since $j_1 \not\in \Hs$ and $\Hs = \tau(\A)$, this component forms a path starting in $j_1$ and ending in a vertex $j_{k+1}$ in $\G\setminus \overline H$;
see Figure~\ref{figure:connectionToSR}.
Let us denote the path as $j_1, \kappa(j_1), j_2, \kappa(j_2), \dots, \kappa(j_{k}), j_{k+1}$
where $j_{t+1} = \tau(\kappa(j_{t}))$ for $t \in [k]$, and $j_t \in \overline H$ for $t \le k$.

\begin{figure}[h]
    \centering
    \includegraphics[width=0.8\textwidth]{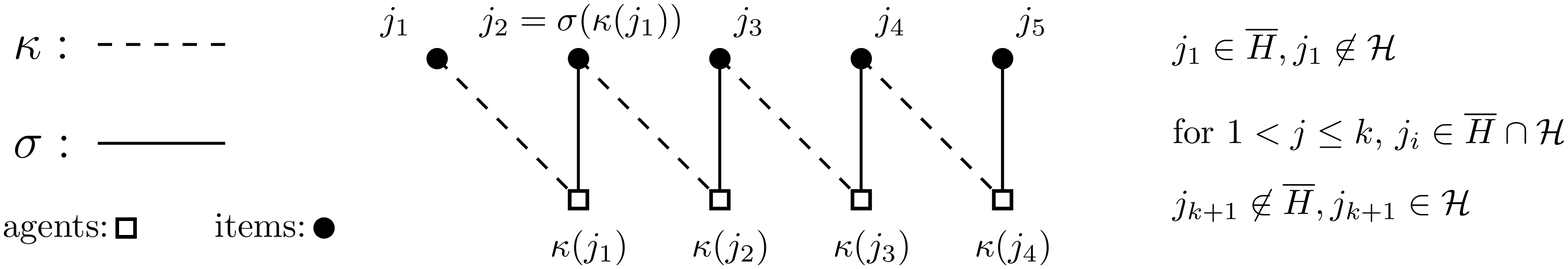}
    \caption{A component of  $\kappa \Delta \tau$ containing $j_1$.}
    \label{figure:connectionToSR}
\end{figure}

Recall, that MBB of each agent is one, therefore $\val_{ij} = p_j$ for each $i,j$ with $x_{ij} > 0$.
By definition of $\kappa$ we have that $\val_{\kappa(j_t) j_t} \ge \val_{\kappa(j_t) j_{t+1}}$ for $t\in[k-1]$.
Moreover, we have $\val_{\kappa(j_1) j_1} = p_{j_1} > 1 \ge p_{j_{k+1}} \ge \val_{\kappa(j_{k}), k+1}$.
Since $j_{t+1} = \tau(\kappa(j_{t}))$, augmenting over the above path will contradict the optimality of $\tau$.
\end{proof}

\section{Conclusions and future work}\label{sec:conclusions}
We have given a constant factor approximation algorithm for the Nash social welfare problem with Rado valuations, assuming that the weights of the agents are bounded by a constant. Rado valuations form a broad subclass of gross substitutes valuations. It remains open to obtain a constant factor approximation for the entire class of gross substitutes valuations, and for even more general classes, such as submodular valuations. The other main open question is to remove the assumption of bounded weights, that is, to obtain a constant factor independent of the parameter $\g$.
 
We note that for subadditive valuations, \citet{barman2020tight} gave an  $O(n)$-approximation and showed that this is essentially tight: an $O(n^{1-\varepsilon})$ approximation would require an exponential number of oracle queries for any fixed $\varepsilon>0$. 

The algorithm is based on a mixed integer programming relaxation, and decomposes into a number of phases. Most reduction steps are applicable for the general subadditive setting. We only require Rado valuations for \ref{phase4}, to obtain an approximate solution with a small support. The factor $\g$ only appears in the reduction in \ref{phase2}, where we restrict each agent to receiving only a single item from the set $\Hs$. Besides extending the result to more general settings, there is
much scope for improving the approximation factor by using tighter analyses and amortizing across the different phases.

For example, we expect that a (mild) extension to budget-Rado valuations should be achievable. 
Similarly to \cite{ChaudhuryCGGHM18,garg2018approximating}, 
this means Rado valuations with a cap on the maximum obtainable value for each agent. 
This only requires a slightly more careful argument in \ref{phase4}.

Our work also highlights Rado valuations as an interesting class of gross substitutes valuations; this could be relevant also for other problems in mechanism design: it is a broad class including most common examples such as weighted matroid rank functions and OXS valuations, yet it has a rich combinatorial structure that can be exploited for algorithm design.

\bibliographystyle{plainnat}
\bibliography{references}
\end{document}